\documentclass[12pt,cls,onecolumn]{IEEEtran}
\usepackage{subfigure,cite,graphicx,psfig,amsmath,amssymb,mathrsfs,epsf}

\begin{document}

\title{Opportunistic Interference Alignment for MIMO Interfering Multiple-Access Channels}
\author{\large Hyun Jong Yang, \emph{Member}, \emph{IEEE}, Won-Yong Shin, \emph{Member}, \emph{IEEE}, \\ Bang Chul Jung, \emph{Member}, \emph{IEEE}, and
Arogyaswami Paulraj, \emph{Fellow}, \emph{IEEE} \\
\thanks{A part of this work was presented in IEEE International Symposium on Information Theory (ISIT), Cambridge, MA, July 2012 \cite{H_Yang12_ISIT}.}
\thanks{H. J. Yang and A. Paulraj are with the Department of Electrical Engineering, Stanford University, Stanford, CA 94305 (email: {hjdbell, apaulraj}@stanford.edu).}
\thanks{W.-Y. Shin is with the Division of Mobile Systems Engineering, College of International Studies, Dankook University, Yongin 448-701, Republic of Korea (E-mail: wyshin@dankook.ac.kr).}
\thanks{B. C. Jung (corresponding author) is with the Department of Information and Communication Engineering, Gyeongsang National University, Tongyeong 650-160, Republic of Korea (E-mail: bcjung@gnu.ac.kr).}
        } \maketitle


\markboth{IEEE Transactions on Wireless Communications, Accepted for
Publication} {Yang {\em et al.}: Opportunistic Interference
Alignment for MIMO Interfering Multiple-Access Channels}


\newtheorem{definition}{Definition}
\newtheorem{theorem}{Theorem}
\newtheorem{lemma}{Lemma}
\newtheorem{example}{Example}
\newtheorem{corollary}{Corollary}
\newtheorem{proposition}{Proposition}
\newtheorem{conjecture}{Conjecture}
\newtheorem{remark}{Remark}

\def \diag{\operatornamewithlimits{diag}}
\def \min{\operatornamewithlimits{min}}
\def \max{\operatornamewithlimits{max}}
\def \log{\operatorname{log}}
\def \max{\operatorname{max}}
\def \rank{\operatorname{rank}}
\def \out{\operatorname{out}}
\def \exp{\operatorname{exp}}
\def \arg{\operatorname{arg}}
\def \E{\operatorname{E}}
\def \tr{\operatorname{tr}}
\def \SNR{\operatorname{SNR}}
\def \dB{\operatorname{dB}}
\def \ln{\operatorname{ln}}

\def \be {\begin{eqnarray}}
\def \ee {\end{eqnarray}}
\def \ben {\begin{eqnarray*}}
\def \een {\end{eqnarray*}}

\begin{abstract}
We consider the $K$-cell multiple-input multiple-output (MIMO)
interfering multiple-access channel (IMAC) with time-invariant
channel coefficients, where each cell consists of a base station
(BS) with $M$ antennas and $N$ users having $L$ antennas each. In
this paper, we propose two opportunistic interference alignment
(OIA) techniques utilizing multiple transmit antennas at each user:
antenna selection-based OIA and singular value decomposition
(SVD)-based OIA. Their performance is analyzed in terms of
\textit{user scaling law} required to achieve $KS$
degrees-of-freedom~(DoF), where $S(\le M)$ denotes the number of
simultaneously transmitting users per cell. We assume that each
selected user transmits a single data stream at each time-slot.
It is shown that the antenna selection-based OIA does not
fundamentally change the user scaling condition if $L$ is fixed,
compared with the single-input multiple-output~(SIMO) IMAC case,
which is given by $\text{SNR}^{(K-1)S}$, where SNR denotes the
signal-to-noise ratio. In addition, we show that the SVD-based OIA
can greatly reduce the user scaling condition to
$\text{SNR}^{(K-1)S-L+1}$ through optimizing a weight vector at each
user. Simulation results validate the derived scaling laws of the
proposed OIA techniques. The sum-rate performance of the proposed
OIA techniques is compared with the conventional techniques in MIMO
IMAC channels and it is shown that the proposed OIA techniques
outperform the conventional techniques.
\end{abstract}


\newpage


\section{Introduction}

Interference management is a crucial problem in wireless
communications. Over the past decade, there has been a great deal of
research to characterize the asymptotic capacity inner-bounds of
interference channels (ICs) using the simple notion of
degrees-of-freedom (DoF), also known as multiplexing gain. Recently,
interference alignment (IA) \cite{V_Cadambe08_TIT,
S_Jafar08_TIT,M_MaddahAli08_TIT,
K_Gomadam11_TIT,C_Suh08_Allerton,A_Motahari09_ArXiv,T_Gou10_TIT,
V_Cadambe08_ISIT, V_Cadambe09_TIT} has emerged as a fundamental
solution to achieve the optimal degrees-of-freedom
(DoF)\footnote{The \textit{optimal} DoF denotes the maximum
achievable DoF for given channel, which is proved by the converse
proof. } in several IC models.
The conventional IA technique for  the $K$-user IC  \cite{V_Cadambe08_TIT} and the $K$-user X channel \cite{V_Cadambe08_ISIT, V_Cadambe09_TIT} is based on several strict conditions as follows. 
Time, frequency, or space domain extension is required to render the
channel model multi-dimensional. To this end, channel randomness,
i.e., time-varying or frequency-selective channel coefficients, is
needed.
Moreover, an arbitrarily large size of the dimension extension is
needed for $K$ greater than 3,
 which results in an excessive bandwidth usage  is required for the decoding of one signal block \cite{C_Suh08_Allerton}. In addition, global channel state information (CSI) is needed at all nodes \cite{V_Cadambe08_TIT, T_Gou10_TIT, T_Gou08_Asilomar,S_Jafar08_TIT,B_Nazer09_ISIT}.


For the interfering multiple-access channel (IMAC) consisting of $K$
cells, where each cell is composed of $N$ users and a single base
station (BS),
Suh and Tse developed a new IA scheme to characterize the DoF achievability of the $K$-cell IMAC \cite{C_Suh08_Allerton} allowing the rank of the interference space to be larger than one. 
The underlying idea of the IA is to align the interference to the
desired interference spaces at the receivers by exploiting
diversity (i.e., randomness) in any resource domain. The scheme
proposed in \cite{C_Suh08_Allerton} utilized the user domain
resource for the IA in the IMAC. This IA scheme based on the user
diversity leads to two interesting results. First, the DoF of the
interference-free network, given by $K$,  can be achieved as $N$
increases.
Second, the size of the time/frequency domain extension is greatly
reduced. Specifically, the finite size of the extension is given by
$O(N)$, which is sufficient  to operate for given $N$. However,
arbitrarily large $N$ is needed to attain $K$ DoF, which results in
an infinite dimension extension in the end. Thus, time-varying or
frequency-selective fading is still required for this scheme.


 Recently, the concept
of opportunistic interference alignment (OIA) was introduced in
\cite{B_Jung11_CL, B_Jung10_Asilomar, B_Jung11_Asilomar,
B_Jung11_TC, B_Jung12_Asilomar}, for the $K$-cell $N$-user
single-input multiple-output (SIMO) IMAC with time-invariant channel
coefficients, where each base station (BS) has $M$ antennas.
In the OIA technique, opportunistic user scheduling is combined with
the spatial domain IA to align the interference to predefined
interference spaces at each BS by exploiting multiuser diversity.
Although several studies independently addressed some of the
aforementioned practical problems of the conventional IA technique
\cite{C_Suh08_Allerton, K_Comadam08_GLOBECOM,
C_Yetis10_TIT,H_Sung10_TWC}, the OIA technique resolved these
practical  issues simultaneously.
The OIA scheme employs the spatial domain IA only with the aid of
opportunistic user scheduling and thus operates with a single
snapshot without any dimension extension.
The purpose of the OIA-related work \cite{B_Jung11_CL,
B_Jung11_Asilomar, B_Jung11_TC} is not only to maximize the DoF as
in the conventional schemes, but also to characterize the trade-off
between the achievable DoF and the number of users required.
It was shown in \cite{B_Jung11_TC} that the OIA scheme achieves $KS$
DoF if $N$ scales faster than $\textrm{SNR}^{(K-1)S}$ in a high SNR
regime, where $S (\le M)$ is the number of selected users in each
cell.


In this paper, we introduce an OIA for the $K$-cell MIMO IMAC with
time-invariant channel coefficients, where each cell consists of one
BS with $M$ antennas and $N$ users having $L$ antennas each.
Inheriting the basic OIA principle \cite{B_Jung11_CL}, the proposed
OIA operates with local CSI at the
transmitter\footnote{\label{footnote:local_CSI}In interference
channels, the local CSI at the transmitter denotes the information
of the channels from the transmitter to all receivers, i.e., its own
transmit links \cite{K_Gomadam11_TIT}. }, no inter-user or intercell
coordination (i.e., distributed scheduling metric calculation), no
dimension extension, and no iterative processing. In
\cite{T_Kim11_arXiv}, the outer bound on the DoF of the MIMO IMAC
with time-invariant channel coefficients was  characterized, and
necessary conditions for $M$ and $L$ needed to achieve the optimal
DoF were derived with global CSI at all nodes.
However, the main goal of the proposed OIA is to characterize a trade-off between the achievable DoF and the number of users required in the MIMO IMAC with arbitrary $M$ and $L$. 
That is, the focus is on studying the user scaling law needed to
achieve the target DoF, given by $KS$, which is optimal if $S=M$.
Scaling conditions required to achieve target performance have a
great impact in providing the convergence rate to the target
performance with respect to considered system parameters, thus
yielding an intuitive performance measure. For instance, it is
common in MIMO systems to evaluate limited feedback schemes by
analyzing the relationship between the codebook size scaling and the
rate-loss \cite{B_Mondal06_TSP,T_Yoo07_JSAC}, and the concept has
been applied also to MIMO ICs \cite{J_Thukral09_ISIT,
R_Krishnamachari10_ISIT}.

In the downlink cellular IC, user scaling laws were developed  for
the OIA \cite{T_Kim11_arXiv,S_Pereira07_Asilomar} and for the
opportunistic interference management with limited feedback
\cite{L_Dritsoula10_SPAWC,Z_Wang09_MILCOM}. These schemes cannot be
easily extended to the IMAC, because there exists a mismatch between
generating interferences at each user and interferences suffered by
each BS from multiple users, thus yielding the difficulty of user
scheduling design.



More specifically, we propose the following two types of OIA:
antenna selection-based OIA and singular value decomposition
(SVD)-based OIA. We then derive the scaling law for required $N$
with respect to SNR, under which $KS$ DoF can be achieved. In the
proposed schemes, each selected user employs transmit beamforming to
mitigate the leakage of interference (LIF) it generates. While the
alignment was performed only through user scheduling in the SIMO
case, the transmit beamforming is used for the MIMO OIA to perform
the spatial domain IA along with opportunistic user scheduling.
Moreover, the additional effort for the feedback of the weight
vector from each selected user to the corresponding BS is in general
required compared to the SIMO case, except for the proposed antenna
selection-based OIA.

We show that for the antenna selection-based OIA, where the best
transmit antenna is selected at each user, required $N$ scales as
$L^{-1}\textrm{SNR}^{(K-1)S}$. Thus, the user scaling condition with
respect to SNR does not fundamentally change, compared with the SIMO
IMAC case~\cite{B_Jung11_TC}, if $L$ is a constant independent of
$N$. However, the sum-rate gain of the antenna selection-based OIA
over the SIMO OIA increases as $L$ grows, whereas no additional
feedback is required. For the SVD-based OIA, each user designs the
weight vector that minimizes the leakage of interference (LIF) using
SVD-based beamforming. We show that the SVD-based OIA can greatly
reduce the user scaling condition to $\textrm{SNR}^{(K-1)S-L+1}$
with the help of the high-rate feedback.
Our schemes are compared with the existing IA schemes for multiuser
ICs, and computer simulations are provided to validate the derived
scaling laws. From this study, besides the fundamental trade-off
between the user scaling condition and the achievable DoF, we
examine that in the MIMO IMAC, there also exists  a trade-off
between the amount of feedback for the weight vectors and the user
scaling condition.

The organization of this paper is as follows.
Section \ref{SEC:system} describes the system and channel models of
MIMO IMAC. The proposed the MIMO OIA scheme is presented in Section
\ref{SEC:OIA}. Both DoF achievability analyses and user scaling laws
are provided in Section \ref{sec:achievability}. The proposed scheme
is compared with the existing MIMO uplink schemes as well as the
converse proof in Section \ref{sec:comp}. Section \ref{SEC:Sim}
provides simulation results and Section \ref{SEC:Conc} concludes the
paper.

\textit{Notations:} $\mathbb{C}$ indicates the field of complex
numbers.  $(\cdot)^{T}$ and $(\cdot)^{H}$ denote the transpose and
the conjugate transpose, respectively.


\section{System and Channel Models} \label{SEC:system}

Let us consider the time-division duplex (TDD) $K$-cell MIMO IMAC,
as depicted in Fig. \ref{fig_sys}. Each cell consists of a BS with
$M$ antennas and $N$ users, each with $L$ antennas.  The number of
users selected to transmit uplink signals in each cell is denoted by
$S\le M$. It is assumed that each selected user transmits a single
spatial stream. To consider nontrivial cases, we assume that $L <
(K-1)S +1$, because all the inter-cell interference can be
completely canceled at the transmitters otherwise\footnote{The case
where $L \ge (K-1)S +1$ and where each selected user transmits
multiple spatial streams is discussed at the end of Section
\ref{subsec:LIF_beamforming} (see Remark \ref{remark:SVD_OIA_null})
and also in Section \ref{SEC:comp_existing} with the comparison to
the existing schemes.}. The channel matrix from user $j$  in the
$i$-th cell to BS $k$ (in the $k$-th cell) is denoted by
$\mathbf{H}_{k}^{[i,j]}\in \mathbb{C}^{M \times L}$, where $i,k\in
\mathcal{K} \triangleq \{ 1, \ldots, K\}$ and $j \in \mathcal{N}
\triangleq \{1, \ldots, N\}$. Time-invariant frequency-flat fading
is assumed, i.e., channel coefficients are constant during a
transmission block, and channel reciprocity between uplink and
downlink channels is assumed. From pilot signals sent from all the
BSs, user $j$ in the $i$-th cell can estimate the channels
$\mathbf{H}_{k}^{[i,j]}$, $k=1, \ldots, K$, utilizing the channel
reciprocity, i.e., the local CSI at the transmitter.
Without loss of generality, the indices of selected users in every
cell are assumed to be $(1, \ldots, S)$. The total DoF are defined
by
\begin{equation}
\textrm{DoF} = \lim_{\textrm{SNR} \rightarrow \infty}
\frac{\sum_{i=1}^{K}\sum_{j=1}^{S}R^{[i,j]}}{\log \textrm{SNR}},
\end{equation}
where $R^{[i,j]}$ denotes the achievable rate for user $j$ in the
$i$-th cell.

\section{Proposed OIA for MIMO IMAC} \label{SEC:OIA}
 We first describe the overall procedure of the proposed OIA scheme for MIMO IMAC, and then derive the achievable sum-rate and present the geometric interpretation of the proposed scheme.
\subsection{Overall Procedure} \label{subsec:overall}
\subsubsection{Initialization (Reference Basis Broadcast)}
 The interference space for the interference alignment at the $k$-th cell is denoted by $\mathbf{Q}_k = \left[ \mathbf{q}_{k,1}, \ldots, \mathbf{q}_{k,M-S}\right]$, where $\mathbf{q}_{k,m} \in \mathbb{C}^{M \times 1}$ is the orthonormal basis, $k\in \mathcal{K}$, $m =1, \ldots, M-S$. BS $k$ independently generates $\mathbf{q}_{k,m}$ from the isotropic distribution over the $M$-dimensional unit sphere.  For given $\mathbf{Q}_k$, BS $k$ also calculates the null space of $\mathbf{Q}_k$, defined by
 \begin{equation}
 \mathbf{U}_k = \left[ \mathbf{u}_{k,1}, \ldots, \mathbf{u}_{k,S} \right] \triangleq \textrm{null}(\mathbf{Q}_k),
 \end{equation}
 where $\mathbf{u}_{k,i}\in \mathbb{C}^{M \times 1}$ is the orthonormal basis, and broadcasts it to all users prior to the communication. The interference basis $\mathbf{Q}_k$ can be chosen arbitrarily such that $\mathbf{Q}_k$ is full rank. A simple way to maximize the performance of the ZF equalization at the BS, which will be discussed in the sequel,  would be choosing $M-S$ columns of the left or right singular matrix of any $M \times M$ matrix as $\mathbf{Q}_k$ and choosing the rest of the $S$ columns as $\mathbf{U}_k$. If $S=M$, then $\mathbf{U}_k$ can be any orthogonal matrix.
 Note that the calculation and broadcast of $\mathbf{U}_k$ is required only once prior to the communication as $\mathbf{Q}_k$ is determined only by $M$ and $S$.

\subsubsection{Stage 1 (Weight Design and Scheduling Metric Feedback)} Let us define  the unit-norm weight vector at user $j$ in the $i$-th cell by $\mathbf{w}^{[i,j]}$, i.e., $\left\| \mathbf{w}^{[i,j]} \right\|^2 = 1$. Two different methods to design $\mathbf{w}^{[i,j]}$ shall be presented in Section \ref{sec:achievability} along with the corresponding user scaling law.
From the notion of $\mathbf{U}_k$ and $\mathbf{H}^{[i,j]}_{k}$, user
$j$ in the $i$-th cell calculates its LIF, which is received at BS
$k$ and not aligned at the interference space $\mathbf{Q}_k$, from
\begin{align}
\tilde{\eta}^{[i,j]}_{k} &= \left\|\textrm{Proj}_{\bot \mathbf{Q}_k}\left( \mathbf{H}_{k}^{[i,j]}\mathbf{w}^{[i,j]}\right)\right\|^2\\
\label{eq:eta_tilde}&= \left\|\mathbf{U}_k^{H}\mathbf{H}_{k}^{[i,j]}
\mathbf{w}^{[i,j]} \right\|^2,
\end{align}
where $i\in \mathcal{K}$, $j \in \mathcal{N} $, and $k\in
\mathcal{K}\setminus i= \{1, \ldots, i-1, i+1, \ldots, K\}$. The
scheduling metric of user $j$ in the $i$-th cell, denoted by
$\eta^{[i,j]}$, is defined by the sum of LIFs, which are not aligned
to the interference spaces at neighboring cells. That is,
\begin{align} \label{eq:eta}
\eta^{[i,j]} &= \sum_{k=1, k\neq i}^{K} \tilde{\eta}^{[i,j]}_{k}.
\end{align}
All the users report their LIF metrics to corresponding BSs.
\subsubsection{Stage 2 (User Selection)}
Upon receiving $N$ users' scheduling metrics in the serving cell,
each BS selects $S$ users having smallest LIF metrics. Note again
that we assume without loss of generality that user $j$, $j=1,
\ldots, S$, in each cell have the smallest LIF metrics and thus are
selected. Subsequently, user $j$ in the $i$-th cell forwards the
information on $\mathbf{w}^{[i,j]}$ to BS $i$ for coherent decoding.

\subsubsection{Stage 3 (Uplink Communication)} The transmit signal vector at user $j$ in the $i$-th cell is given by $\mathbf{w}^{[i,j]}x^{[i,j]}$, where $x^{[i,j]}$ is the transmit symbol with unit average power, and the received signal at BS $i$ can be written as:
\begin{align}
\mathbf{y}_i &= \underbrace{\sum_{j=1}^{S}\mathbf{H}_{i}^{[i,j]} \mathbf{w}^{[i,j]}x^{[i,j]}}_{\textrm{desired signal}} \nonumber \\
& \hspace{30pt}+ \underbrace{\sum_{k=1, k\neq i}^{K} \sum_{m=1}^{S}
\mathbf{H}_{i}^{[k,m]}\mathbf{w}^{[k,m]}x^{[k,m]}}_{\textrm{inter-cell
interference}} + \mathbf{z}_i,
\end{align}
  where $\mathbf{z}_i \in \mathbb{C}^{M \times 1}$ denotes  the additive noise,  each element of which is independent and identically distributed complex Gaussian with zero mean and the variance of $\textrm{SNR}^{-1}$.
As in SIMO IMAC~\cite{B_Jung11_CL,B_Jung11_TC}, the linear
zero-forcing (ZF) detection is applied at the BSs to null inter-user
interference for the home cell users' signals.
From the notion of $\mathbf{H}^{[i,j]}_{i}$ and
$\mathbf{w}^{[i,j]}$,  BS $i$ obtains the sufficient statistics for
parallel decoding
\begin{equation} \label{eq:r_i}
\mathbf{r}_i = \left[r_{i,1}, \ldots, r_{i,S}
\right]^{\textrm{T}}\triangleq
{\mathbf{F}_i}^{H}\mathbf{U}_i^{H}\mathbf{y}_i,
\end{equation}
where $\mathbf{U}_i$ is multiplied to remove the inter-cell
interference components that are aligned at the interference space
of BS $i$, $\mathbf{Q}_i$, and $\mathbf{F}_{i}\in \mathbb{C}^{S
\times S}$ is the ZF equalizer defined by
\begin{align}
\mathbf{F}_i&= \left[ {\mathbf{f}_{i,1}}, \ldots, {\mathbf{f}_{i,S}}\right] \nonumber \\
\label{eq:F_def}&\triangleq
\left(\left[{\mathbf{U}_i}^{H}\mathbf{H}_{i}^{[i,1]}
\mathbf{w}^{[i,1]}, \ldots, {\mathbf{U}_i}^{H}\mathbf{H}_{i}^{[i,S]}
\mathbf{w}^{[i,S]}\right]^{-1}\right)^{H}.
\end{align}

For a comprehensive overview, the overall sequential procedure is
illustrated in Fig. \ref{fig_block}.
Note that we assume low-rate perfect information exchanges for
Stages 1--3, such as feedback of the scheduling metric, broadcast of
user selection information, feedforward of weight vector
information, as in
\cite{N_JindalTIT_06,T_Yoo07_JSAC,J_Thukral09_ISIT,R_Krishnamachari10_ISIT,S_Pereira07_Asilomar,B_Jung11_TC}.

\subsection{Sum-Rate Calculation}\label{subsec:sum_rate}

From (\ref{eq:r_i}), the $j$th spatial stream, $r_{i,j}$, is written
as
\begin{align}
r_{i,j}&= x^{[i,j]}+ \sum_{k=1, k\neq i}^{K} \sum_{m=1}^{S} {\mathbf{f}_{i,j}}^{H}{\mathbf{U}_{i}}^{H}\mathbf{H}_{i}^{[k,m]}\mathbf{w}^{[k,m]}x^{[k,m]} \nonumber \\
&\hspace{80pt}+ {\mathbf{f}_{i,j}}^{H}\mathbf{z}_i^{\prime},
\end{align}
where $\mathbf{z}_i^{\prime} \triangleq
{\mathbf{U}_{i}}^{H}\mathbf{z}_i$. Thus, $R^{[i,j]}$ is given by
\begin{align} \label{eq:rate_general}
R^{[i,j]} &= \log\left( 1+ \textrm{SINR}^{[i,j]}\right) \nonumber \\
&= \log \left( 1+
\frac{\textrm{SNR}}{\left\|\mathbf{f}_{i,j}\right\|^2 +  I_{i,j} }
\right),
\end{align}
where $\textrm{SINR}^{[i,j]}$ denotes the
signal-to-noise-and-interference ratio of the user $j$ in the $i$-th
cell and $I_{i,j}$ is the sum-interference defined by
\begin{equation} \label{eq:SINR_def}
I_{i,j} = \sum_{k=1, k\neq i}^{K} \sum_{m=1}^{S} \left|
{\mathbf{f}_{i,j}}^{H}
{\mathbf{U}_i}^{H}\mathbf{H}_{i}^{[k,m]}\mathbf{w}^{[k,m]} \right|^2
\textrm{SNR}.
\end{equation}




\subsection{Geometric Interpretation}\label{subsec:geometrical}
If $S<M$ and the interference from user $m$ in the $k$-th cell to BS
$i$ is aligned to $\mathbf{Q}_i$, i.e.,
\begin{equation} \label{eq:IA_span}
\mathbf{H}^{[k,m]}_{i}\mathbf{w}^{[k,m]}\in \textrm{span} \left[
\mathbf{Q}_i \right],
\end{equation}
 then it is nulled in $\mathbf{r}_i$ because ${\mathbf{U}_i}^{H} \mathbf{H}^{[k,m]}_{i}\mathbf{w}^{[k,m]} = \mathbf{0}$, i.e., $\tilde{\eta}^{[k,m]}_{i} = 0$.   If $S=M$, the LIF metric is simplified to $\tilde{\eta}^{[i,j]}_{k}= \left\|\mathbf{H}_{k}^{[i,j]} \mathbf{w}^{[i,j]} \right\|^2$. In this case, no IA is conducted and only the opportunistic interference nulling (OIN) is performed as in the OIN for the SIMO IMAC \cite{B_Jung11_TC}. We do not separately describe this OIN mode, as it can be taken into account by the OIA framework.

 Figure \ref{fig_OIA} illustrates the proposed MIMO OIA for $K=2$, $M=3$, and $S=2$. The interference terms $\mathbf{H}^{[1,1]}_{2}\mathbf{w}^{[1,1]}$ and  $\mathbf{H}^{[1,2]}_{2}\mathbf{w}^{[1,2]}$ should be aligned to the interference space $\mathbf{q}_{2,1}$ at BS 2, while we only require for the signal vectors $\mathbf{H}^{[1,1]}_{1}\mathbf{w}^{[1,1]}$ and $\mathbf{H}^{[1,2]}_{1}\mathbf{w}^{[1,2]}$ to be distinguishable at BS 1. Similarly, $\mathbf{H}^{[2,1]}_{1}\mathbf{w}^{[2,1]}$ and  $\mathbf{H}^{[2,2]}_{1}\mathbf{w}^{[2,2]}$ should be aligned to $\mathbf{q}_{1,1}$ at BS 1, while $\mathbf{H}^{[2,1]}_{2}\mathbf{w}^{[2,1]}$ and $\mathbf{H}^{[2,2]}_{2}\mathbf{w}^{[2,2]}$ need to be distinguishable at BS 2.
 The main task of the achievability proof is to show that $\tilde{\eta}^{[i,j]}_{k}$ can be made arbitrarily small for all cross-links through opportunistic scheduling and beamforming, which proves that the IA conditions (\ref{eq:IA_span}) hold true almost surely for all $i\in \mathcal{K}$, $m \in \mathcal{S}\triangleq \{1, \ldots, S\}$, and $k\in\mathcal{K}\setminus i$. Note that for given $\mathbf{w}^{[i,j]}$, the signal vectors at each BS are distinguishable, since the channel coefficients are generated from continuous distributions. Therefore, in such case, the DoF of $KS$ is achievable.

\section{DoF Achievability}\label{sec:achievability}
In this section, we present two different beamforming strategies to
design $\mathbf{w}^{[i,j]}$ at each user, and characterize the DoF
achievability for each strategy in terms of the user scaling law.

\subsection{Antenna Selection}\label{subsec:AS}
In the antenna selection-based OIA, only one transmit antenna is
selected to transmit at each user, i.e., $\mathbf{w}^{[i,j]}\in
\left\{\mathbf{e}_1, \ldots, \mathbf{e}_L \right\}$, where
$\mathbf{e}_l$ denotes the $l$-th column of the $(L \times
L)$-dimensional identity matrix. Let us denote the $l$-th column of
$\mathbf{H}_{k}^{[i,j]}$ by $\mathbf{h}_{k,l}^{[i,j]}$, $l\in \{1,
\ldots, L\}$. Then, user $j$ in the $i$-th cell chooses the optimal
weight vector as $\mathbf{w}^{[i,j]}_{\textrm{AS}} =
\mathbf{e}_{\hat{l}(i,j)}$, where the index $\hat{l}(i,j)$ is
obtained from
\begin{equation}\label{eq:AS_select}
\hat{l}(i,j) =  \arg \min_{1 \le l\le L}  \sum_{k=1, k\neq i}^{K}
\left\|{\mathbf{U}_k}^{H}\mathbf{h}_{k,l}^{[i,j]} \right\|^2.
\end{equation}
Then, the corresponding scheduling metric is given by
\begin{equation}\label{eq:AS_select_LIF_metric}
\eta^{[i,j]}_{\textrm{AS}} = \sum_{k=1, k\neq i}^{K}
\left\|{\mathbf{U}_k}^{H}\mathbf{h}_{k,\hat{l}(i,j)}^{[i,j]}
\right\|^2
\end{equation}
 and is reported to BS $i$. \pagebreak[0]
Since the $\hat{l}(i,j)$-th column of the channel matrix,
$\mathbf{h}_{i,\hat{l}(i,j)}^{[i,j]}$, is the effective channel
vector at BS $i$, the feedback is not needed if user $j$ in the
$i$-th cell transmits the uplink pilot to BS $i$ only through the
$\hat{l}(i,j)$-th antenna after it is selected to
transmit.\pagebreak[0]

The following theorem establishes the DoF achievability of the
antenna selection-based OIA.
\begin{theorem}[User scaling law: Antenna selection-based OIA] \label{theorem:AS}
The antenna selection-based OIA with the scheduling metric
(\ref{eq:AS_select_LIF_metric}) achieves
\begin{equation}
\textrm{DoF} \ge KS
\end{equation}
with high probability   if
\begin{equation} \label{EQ:N_scaling_AS}
N=\omega\left(L^{-1}\textrm{SNR}^{(K-1)S}\right),
\end{equation}
where a function $f(x)$ defined by $f(x) = \omega(g(x))$ implies
that $\lim_{x \rightarrow \infty} \frac{g(x)}{f(x)}=0$.
\end{theorem}


\begin{proof}
See Appendix \ref{app:AS_theorem}.
\end{proof}
Note that in the SIMO IMAC, the OIA scheme achieves the DoF of $KS$
if $N = \omega\left(\textrm{SNR}^{(K-1)S} \right)$ \cite[Theorem
1]{B_Jung11_TC}. Thus, the antenna selection-based OIA does not
fundamentally change the user scaling if $L$ is fixed. Note that
however, the user scaling condition is reduced even without any
additional feedback, compared to the SIMO case, if $L$ scales with
respect to SNR. The following remark discusses the cooperative
feature the opportunistic gain obtained from the user and antenna
diversity in the antenna selection-based OIA.
\begin{remark} \label{remark:AS_cooperative}
If $L$ scales faster than $\textrm{SNR}^{\psi_L}$, where $\psi_L$ is
a positive scalar, then the user scaling condition to achieve the
DoF of $KS$ is given by $N = \omega\left(
\textrm{SNR}^{(K-1)S-\psi_L}\right)$.
 If $\psi_L = (K-1)S$, then the DoF of $KS$ is obtained with high probability for any $N\ge S$. In such case, the opportunistic gain is sufficiently obtained only through the antenna diversity. In other words, the opportunistic gain can be achieved cooperatively from the user and antenna diversity.
\end{remark}

Now as a corollary to Theorem 1 in \cite{B_Jung11_TC}, we discuss
the upper-bound on the user scaling law with the antenna selection
by considering the general case where more than one transmit spatial
stream are allowed at each user.
\begin{corollary} \label{corollary:general_AS_OIA}
Suppose that user $j$ in the $i$-th cell selects $n^{[i,j]}$
transmit antennas with smaller LIF metrics, where the $l$-th
antenna's LIF metric is given by $\sum_{k=1, k\neq i}^{K}
\left\|{\mathbf{U}_k}^{H}\mathbf{h}_{k,l}^{[i,j]} \right\|^2$, $l\in
\{1, \ldots, L\}$. Then, the \textit{general antenna selection-based
OIA}, in which BS $i$ selects $S_i$ users with smaller sum-LIF
metrics, achieves $KS$ DoF with high probability if $N =
\omega\left(L^{-1}\textrm{SNR}^{(K-1)S}\right)$, and if $n^{[i,j]}$
and $S_i$ are chosen such that $S = \sum_{j=1}^{S_i} n^{[i,j]}$,
$i\in \mathcal{K}$, and such that the sum-LIF of the selected $S$
spatial channels is minimized at each cell.
\end{corollary}
\begin{proof}
Since the considered scheme is equivalent to selecting  $S$ spatial
channels (transmit antennas) with smaller LIF metrics amongst $NL$
spatial channels, which is also equivalent to the SIMO OIA with $NL$
users, the proof is immediate from \cite[Theorem 1]{B_Jung11_TC}.
\end{proof}
\begin{remark} \label{remark:AS_comparison}
In the general antenna selection approach, the optimal $n^{[i,j]}$
should be determined to find the best $S$ spatial channels, which in
general requires a joint optimization with global CSI or $L$ times
increased the feedback phases for each user to feed back all
individual LIF metrics for $L$ antennas. Surprisingly, Theorem
\ref{theorem:AS} indicates that the antenna selection-based OIA with
single spatial stream at each user is enough to achieve the same
result, in which the scheduling metric is calculated at each user
using local CSI without any cooperation or additional feedback. It
is more surprising that the selection of the best one out of the $L$
spatial channels at each user does not degrade the diversity gain in
terms of the user scaling law, compared to the selection of the best
$S$ out of $NL$ spatial channels. The result is encouraging, since
we can expect the same benefit of increasing $N$ to $NL$ by simply
increasing the number of antennas at the users.
\end{remark}

\subsection{SVD-Based OIA} \label{subsec:LIF_beamforming}
In the SVD-based OIA, each user finds the optimal weight vector that
minimizes its LIF metric. The same beamforming technique was also
considered in \cite{L_Wang11_GLOBECOM,J_Tang11_IMSAA} for the MIMO
IMAC, however, our focus is to derive a user scaling law and thereby
to analytically examine the relationship between the number of users
and the beamforming techniques used.

The LIF metric for the SVD-based OIA is defined by
\begin{align}
\eta^{[i,j]}_{\textrm{SVD}} &= \sum_{k=1, k\neq i}^{K} \left\|
{\mathbf{U}_k}^{H}\mathbf{H}_{k}^{[i,j]}\mathbf{w}^{[i,j]}\right\|^2
= \left\| \mathbf{G}^{[i,j]} \mathbf{w}^{[i,j]}\right\|^2,
\end{align}
where $\mathbf{G}^{[i,j]}\in \mathbb{C}^{(K-1)S\times L}$ is the
stacked cross-link channel matrix, defined by
\begin{align} \label{eq:G_def}
\mathbf{G}^{[i,j]} &\triangleq \Bigg[ \left({\mathbf{U}_1}^{H}\mathbf{H}_{1}^{[i,j]}\right)^{T}, \ldots, \left({\mathbf{U}_{i-1}}^{H}\mathbf{H}_{i-1}^{[i,j]}\right)^{T},  \nonumber\\
&
\hspace{20pt}\left({\mathbf{U}_{i+1}}^{H}\mathbf{H}_{i+1}^{[i,j]}\right)^{T},
\ldots, \left({\mathbf{U}_K}^{H}\mathbf{H}_{K}^{[i,j]}\right)^{T}
\Bigg]^{T}.
\end{align}
Let us denote the SVD of $\mathbf{G}^{[i,j]}$ as
\begin{equation} \label{eq:G_SVD}
\mathbf{G}^{[i,j]} =
\boldsymbol{\Omega}^{[i,j]}\boldsymbol{\Sigma}^{[i,j]}{\mathbf{V}^{[i,j]}}^{H},
\displaybreak[0]
\end{equation}
where $\boldsymbol{\Omega}^{[i,j]}\in \mathbb{C}^{(K-1)S\times L}$
and $\mathbf{V}^{[i,j]}\in \mathbb{C}^{L\times L}$ consist of $L$
orthonormal columns, and $\boldsymbol{\Sigma}^{[i,j]} =
\textrm{diag}\left( \sigma^{[i,j]}_{1}, \ldots,
\sigma^{[i,j]}_{L}\right)$, where $\sigma^{[i,j]}_{1}\ge \cdots
\ge\sigma^{[i,j]}_{L}$. \pagebreak[0] Then, it is apparent that the
optimal $\mathbf{w}^{[i,j]}$ is determined as
\begin{equation} \label{eq:W_SVD}
\mathbf{w}^{[i,j]}_{\textrm{SVD}} = \arg  \min_{\mathbf{v}} \left\|
\mathbf{G}^{[i,j]} \mathbf{v}\right\|^2 =\mathbf{v}^{[i,j]}_{L},
\end{equation}
where $\mathbf{v}^{[i,j]}_{L}$ is the $L$-th column of
$\mathbf{V}^{[i,j]}$. With this choice the LIF metric is simplified
to
\begin{equation} \label{eq:LIF_beamforming_simple}
\eta^{[i,j]}_{\textrm{SVD}} = {\sigma^{[i,j]}_{L}}^2.
\end{equation}
All the users report their LIF metrics to the corresponding BSs and
BS $i$ selects $S$ users with smaller $\eta^{[i,j]}_{\textrm{SVD}}$
values among $N$ users than the rest. To construct $\mathbf{F}_i$
defined in (\ref{eq:F_def}) at BS $i$ for given selected user $j$,
$i=1, \ldots, K$, $j=1, \ldots, S$,  the information of
$\mathbf{w}^{[i,j]}_{\textrm{SVD}}$ needs to be known by BS $i$
through the feedback with a sufficiently high rate.

At this point, we introduce a useful lemma for the polynomial
expression of the CDF of $\eta^{[i,j]}_{\textrm{SVD}}$.
\begin{lemma}[CDF of $\eta^{[i,j]}_{\textrm{SVD}}$] \label{lemma:F_phi} \label{line:lemma:F_phi}
The CDF of $\eta^{[i,j]}_{\textrm{SVD}}$, denoted by
$F_{\sigma}(x)$, can be written as
\begin{equation} \label{eq:F_phi}
F_{\sigma}(x) = \alpha x^{(K-1)S-L+1} + o\left(x^{(K-1)S-L+1}
\right),
\end{equation}
for $0 \le x <1$, where $\alpha$ is a constant determined by $K$,
$S$, and $L$.
\end{lemma}

\begin{proof}
See Appendix \ref{app:CDF_SVD}.
\end{proof}

Now the following theorem establishes the DoF achievability of the
SVD-based OIA.

\begin{theorem}[User scaling law: SVD-based OIA] \label{theorem:BF}
The proposed SVD-based OIA scheme with the scheduling metric
(\ref{eq:LIF_beamforming_simple}) achieves
\begin{equation}
\textrm{DoF} \ge KS
\end{equation}
with high probability  if
\begin{equation} \label{EQ:N_scaling_BF}
N=\omega\left(\textrm{SNR}^{(K-1)S-L+1}\right).
\end{equation}
\end{theorem}

\begin{proof}
See Appendix \ref{app:BF_theorem}.
\end{proof}
Therefore, unlike the antenna selection, the SVD-based OIA
fundamentally lowers the power of the SNR scaling condition required
to achieve the DoF of $KS$. Note that however, this reduced scaling
is achieved at the cost of the sufficiently high-rate feedback of
$\mathbf{w}_{\textrm{SVD}}^{[i,j]}$ from all the selected users to
associated BSs. Noting that the antenna selection-based OIA needs no
feedback, the antenna selection- and SVD-based OIA schemes are the
two extremes of the trade-off between the feedback amount and the
user scaling condition to achieve the DoF of $KS$.


The following remark discusses the trivial case of the SVD-based OIA
in terms of the antenna configuration, where the inter-cell
interference is perfectly canceled only through transmit
beamforming.
\begin{remark} \label{remark:SVD_OIA_null}
Note that if $L\ge (K-1)S +1$, then
$\mathbf{G}^{[i,j]}\in\mathbb{C}^{(K-1)S \times L }$ in
(\ref{eq:G_def}) becomes a wide matrix and the singular value
corresponding to $\mathbf{v}_{L}^{[i,j]}$ is 0. Therefore,
$\mathbf{w}^{[i,j]}$ can be chosen such that
$\eta^{[i,j]}_{\textrm{SVD}}=0$. The result is intuitively immediate
because the total rank of the effective interfering channels from
each user to neighboring cells is $(K-1)S$ and because at least one
additional rank is required for each user to transmit one data
stream. From this result, it can be easily seen that in the case
where each selected user transmits $a\le M$ data streams, all the
inter-cell interference will be canceled through the SVD-based OIA
if $L \ge (K-1)Sa + a$. In such case, the number of selected users,
$S$, should be equal to or lower than $\lfloor \frac{M}{a} \rfloor$.
\end{remark}


%
\subsection{User Scaling Laws for Cell-Dependent $L$, $M$, $N$, and  $S$} \label{subsec:different_parameters}
In this subsection, we examine the user scaling laws for the case
where $L$, $M$, $N$, and $S$ are different from cells. Let us denote
these parameters at the $i$-th cell by $L_i$, $M_i$, $N_i$, and
$S_i$, respectively. The following theorem establishes the user
scaling laws under this scenario.
\begin{theorem}\label{th:different_parameters}
With the cell-dependent parameters, the antenna selection- and
SVD-based OIA schemes achieve $KS$ DoF with high probability if
\begin{equation} \label{eq:DP:user_scaling}
N_i = \omega\left( \textrm{SNR}^{S^{\prime}_{i}K} \right),
\hspace{2pt} \textrm{and}\hspace{5pt} N_i = \omega \left(
\textrm{SNR}^{S^{\prime}_{i}K - L_i +1}\right),
\end{equation}
respectively, where $S^{\prime}_{i} = \sum_{k\neq i, k=1}^{K} S_k$.
\end{theorem}

\begin{proof}
See Appendix \ref{app:cell_dependent}.
\end{proof}

From Theorem \ref{th:different_parameters}, it is seen that growing
the number of serving users at the $i$-th cell, $S_i$, increases the
number of users required at all other cells for both the antenna
selection- and SVD-based OIA. This is because increasing $S_i$
implies a reduced rank of the interference space at the $i$-th cell,
on which users from the other cells attempt to align their signals.
For the SVD-based OIA, large $L_i$ reduces the user scaling
condition of only the $i$-th cell.


\section{Comparison with Upper Bounds and Existing Schemes} \label{sec:comp}

In this section, to verify the optimality of the proposed OIA
schemes, we introduce an upper bound on the DoF. We also compare our
schemes with existing schemes in terms of the achievable DoF and the
computational complexity.
\subsection{Upper Bounds for DoF} \label{subsec:DoF_UB}
We now show an upper limit on the DoF in MIMO IMAC and discuss how
to achieve the DoF upper bound.
For completeness, we briefly review Corollary 1 of
\cite{T_Kim11_arXiv} in which the outer bound on the DoF of the MIMO
IMAC  is given by
\begin{align}  \label{eq:Outer_bound_DoF}
\textrm{DoF} \le &\min\bigg\{NKL, KM, \frac{NK\max \left(NL, (K-1)M \right)}{N+K-1},\nonumber \\
& \hspace{70pt}\frac{NK\max\left((K-1)L, M \right) }{N+M-1} \bigg\}.
\end{align}

Now it is shown that choosing $S = M$, the proposed schemes achieve
$KM$ DoF with arbitrarily large $N$ scaling according to
(\ref{EQ:N_scaling_AS}) and (\ref{EQ:N_scaling_BF}). Note again that
with this choice, interference nulling is carried out through
opportunistic user scheduling.
As $N$ increases, the outer bound (\ref{eq:Outer_bound_DoF}) is
reduced to $KM$, and hence, our schemes can asymptotically achieve
the optimal DoF.

\subsection{DoF Comparison with Existing Methods} \label{SEC:comp_existing}
In this subsection, the proposed OIA schemes are compared with the two existing strategies \cite{T_Gou10_TIT,T_Kim11_arXiv} that also achieve the optimal DoF in $K$-cell MIMO uplink networks. 
Let us first consider the $K$-user MIMO IC \cite{T_Gou10_TIT} with
time-invariant or frequency-selective fading, which can be regarded
as a MIMO IMAC with $N=1$. Consider the case where $L>M$.
Then both the scheme in \cite{T_Gou10_TIT} and the proposed
SVD-based OIA with each user transmitting $M$ spatial streams
achieve the optimal DoF, given by $KM$, if $L\ge K M$ \cite[Theorem
3]{T_Gou10_TIT}. Note that in this case, interference can be
perfectly nulled only through SVD-based beamforming and thus no
opportunistic gain is needed (refer to Remark
\ref{remark:SVD_OIA_null}).
The achievable scheme in \cite{T_Gou10_TIT} operates under time-varying or frequency selective fading channels with global CSI at all nodes, and the size of the time/frequency domain extension is given by $(L/M+1)(n+1)^{\Gamma}$, where $\Gamma = KL/M \cdot\left(K-L/M-1\right)$ and $n$ should be arbitrarily large to obtain $KM$ DoF. 
For the $K$-user MIMO IC with time-invariant channel coefficients
\cite{T_Gou10_TIT}, a necessary condition for the parameter $M$ is
also needed to achieve the optimal DoF, which is given by $M\le
(K-2) L$ for $K>4$. Hence, arbitrarily large $M$ is also required as
$K$ increases, whereas our schemes have no necessary condition for
$M$.

Now, let us turn to the $K$-cell MIMO IMAC studied in \cite{T_Kim11_arXiv}. For $K=2$, both the transmit zero forcing scheme in \cite{T_Kim11_arXiv} and the proposed SVD-based OIA with $N=M$ achieves the $2M$ DoF  if $L \ge (K-1)M+1 = M+1$. 
However, for  $K>2$, the scheme in \cite{T_Kim11_arXiv} needs the
necessary condition $M \ge KS$ to obtain $KS$ DoF  \cite[Theorem
3]{T_Kim11_arXiv}, which is not needed in the proposed OIA schemes.
Moreover, the precoding matrices are designed based on the notion of
global CSI in \cite{T_Kim11_arXiv}.

\subsection{Computational Complexity} \label{subsec:complexity}
In this subsection, we briefly discuss the computational complexity
of the two proposed schemes and compare it to the complexity of the
SISO IMAC scheme, Suh and Tse's scheme.
The computational effort is analyzed in two-fold: the user
computation and the BS computation.  We omit the analysis of the
detection and decoding complexity after the equalization at each BS,
since it is all the same for the schemes considered.

\subsubsection{Antenna selection-based OIA}
Each user calculates (\ref{eq:AS_select}), from which the scheduling
metric (\ref{eq:AS_select_LIF_metric}) can also be obtained. From
the results of \cite{E_Chu99_Book,Convex_Optimization}, it can be
easily shown that the calculation of (\ref{eq:AS_select}) requires
$8(K-1)MLS + 6(K-1)LS -2L$ floating point operations (flops), real
additions or multiplications; thus, the complexity can be denoted by
$O(KLMS)$.

Upon receiving $N$ scheduling metrics, each BS selects $S$ users
with smaller scheduling metrics out of $N$ users, which can be
performed with linear-time complexity, i.e., $O(N)$, by the partial
sorting algorithm \cite{J_Chambers71_ACM}. Next, the construction of
the effective channel matrix, i.e.,  ${\mathbf{F}_i}^{-H}$ (See
(\ref{eq:F_def})),  requires  $8MS^2-2S^2$ flops. The inversion of
this effective channel matrix to get $\mathbf{F}_i$ needs $O(S^3)$
flops, and the calculation of $\mathbf{r}_i$ given in (\ref{eq:r_i})
requires $8MS+8S^2-4S$ flops. Therefore, noting that $S\le M$, the
overall computational complexity at each BS is $O(N+MS^2)$.

\subsubsection{SVD-based OIA}
Each user first constructs $\mathbf{G}^{[i,j]}$ defined in
(\ref{eq:G_def}), which requires $8(K-1)MLS-2(K-1)LS$ flops, i.e.,
$O(KLMS)$. Note that the weight vector and the scheduling metric can
be simultaneously obtained from the SVD of $\mathbf{G}^{[i,j]}$. The
efficient and precise SVD method based on the Householder
reflections and the QR decomposition can be performed with
$O\left(KSL^2\right)$ flops \cite{G_Golub_MATRIX_COMPUTATION}.
Consequently, the computational complexity of the SVD-based OIA at
each user is $O(KSL^2 + KLMS)$.

All the procedure at each BS is the same as that of the antenna
selection-based OIA except the construction of the effective channel
matrix, which requires  $8MLS + 8MS^2 - 2MS - 2S^2$ flops.
Consequently, the overall complexity  at each BS is given by $O(N +
MLS + MS^2)$.

Table I summarizes the computational complexity of the OIA schemes
with the comparison to the SIMO case. It is obvious that the
complexity is the lowest for the SIMO OIA and is the highest for the
SVD-based OIA. It is seen that as $L$ increases, the complexity
difference between the three schemes becomes greater.

\subsubsection{SISO IMAC} \label{subsubsec:SISO_IMAC}
Now we briefly discuss the computational complexity of Suh and Tse's
scheme \cite{C_Suh08_Allerton}. Since this scheme applies only to
the SISO IMAC, the comparison to this scheme is to roughly show the
computational efficiency of the proposed  schemes.
Each user in Suh and Tse's scheme finds the inversions of $K-1$ $(n
\times n)$-dimensional matrices and the Kronecker multiplications of
$K-1$ $n$-dimensional vectors, where $n = \sqrt[K-1]{N}+1$. This
calculation at each user requires $O(N+K\sqrt[K-1]{N}^{3})$ flops.
Another heavy calculation in this scheme is to find the
$(K-1)$-level decompositions of $n^{K-1}\times n^{K-1}$ matrices,
which cannot be systematically performed. In addition, the
complexity for the equalization at each BS is dominated by the
effort to find the inversion of an $(n^{K-1} \times
n^{K-1})$-dimensional matrix, which needs $O(N^3)$ flops.

Considering the fact that Suh and Tse's scheme requires much lower
dimension extension size than the conventional IA schemes and thus
is already computationally attractive, the proposed schemes is more
computationally effective compared to the previous schemes. In
addition, it should be stressed that both the dimension extension
size and $N$ need to be arbitrarily large to achieve the optimal DoF
with Suh and Tse's scheme, whilst arbitrarily large  $N$ suffices
the condition for the optimal DoF for the proposed schemes.


\section{Simulation Results} \label{SEC:Sim}

In this section, through computer simulations, we evaluate the sum
of LIF and the sum-rate of the proposed OIA schemes, operating with
finite $N$ and SNR in the MIMO IMAC. The max-SNR scheme is compared,
in which the weight vectors and the scheduling metrics are
calculated at each user in a distributed manner only with local CSI.
Specifically, each user employs eigen-beamforming to maximize its
effective SNR and each BS selects $S$ users having higher effective
SNRs up to the $S$-th largest one. The OIA scheme employing a fixed
weight vector, i.e., $\mathbf{w}^{[i,j]} = \mathbf{e}_1$ for all
users, is also considered, which can be treated as the OIA scheme
for SIMO IMAC. Thus, we refer this scheme as `SIMO OIA'.


Figure \ref{LIF_N} depicts the log-log plot of the sum of LIF,
termed as sum-LIF, i.e., $\sum_{i=1}^{K}\sum_{j=1}^{S}\eta^{[i,j]}$,
versus $N$ when $K=3$, $M=L=2$, and SNR is 10dB.
This performance measurement enables us to measure the quality of
the proposed OIA schemes, as shown in \cite{K_Gomadam11_TIT}.
Specifically, Fig. \ref{LIF_N} exhibits how rapidly the network
becomes an error-free network with respect to $N$. Since the user
selection of the max-SNR scheme does not contribute to the reduction
of the LIF, the sum-LIF of the max-SNR scheme remains constant for
increasing $N$. The sum-LIF of the antenna selection-based OIA
decreases with respect to $N$ at the same rate of the SIMO OIA,
because the antenna selection-based OIA is subject to the user
scaling condition $\textrm{SNR}^{(K-1)S}$ if $L$ is fixed.
On the other hand, the decreasing rate of the SVD-based OIA is
higher, which is subject to the user scaling condition
$\textrm{SNR}^{(K-1)S-L+1}$. As $S$ decreases, the decreasing rates
of both the antenna selection- and SVD-based OIA schemes become
higher due to the lowered scaling conditions.

Figure \ref{LIF_L} shows the log-log plot of the sum-LIF versus $L$
when $K=3$, $M=3$, and $N=100$. For the antenna selection-based OIA,
the sum-LIF decreases linearly in log-log scale. On the other hand,
the sum-LIF of the SVD-based OIA decreases much faster than the
antenna selection-based OIA case and becomes zero if $L\ge
(K-1)S+1=5$ (refer to Remark \ref{remark:SVD_OIA_null}). Note that
however, the feedback redundancy for the weight vectors grows as $L$
increases in the SVD-based OIA, whereas no feedback is required
regardless of $L$ in the antenna selection-based OIA.

Figure \ref{rates_SNR} depicts the sum-rates versus SNR when $K=3$,
and $M=L=2$ for (a) $N=20$ and (b) $N=100$. The sum-rates of the
considered schemes are saturated in the sufficiently high SNR
regime, because the inter-cell interference cannot approach zero for
fixed $N$ values. That is, the SINR will be upper-bounded by a
finite value for all schemes. In fact, $S$ determines the amount of
the interference level as well as the total DoF. For the max-SNR
scheme, the interference at each BS increases as $S$ increases,
whereas the sum-rate is increased by $S$ times. The rate at each BS
is approximately given by $S\log\left(
1+\frac{\textrm{SNR}}{1+(K-1)S \cdot \Delta}\right)$, where $\Delta$
denotes the amount of the interference received from a single user
in a neighboring cell. Since this rate is a monotonically increasing
function of $S$, the rate of the max-SNR scheme grows with $S$. On
the other hand, the proposed schemes can significantly suppress  the
interference. Hence, the cases with $S=2$ show higher sum-rates than
the cases with $S=1$ in the low SNR regime where the noise is
dominant over the interference, and vice versa in the high SNR
regime where it becomes more important to minimize the interference.
As $N$ increases, the interference can be more reduced, and thus the
crossover SNR points, where the sum-rates for the cases $S=1$ and
$S=2$ are identical,  become higher. From Fig.  \ref{rates_SNR}, the
crossover SNR points of the antenna selection-based OIA appear
approximately at 6dB when $N=20$ and at 9.1dB when $N=100$, whereas
those of the SVD-based OIA are 8.1dB when $N=20$ and 12.1dB when
$N=100$.

Figure \ref{rates_N} depicts the sum-rates versus $N$ when $K=3$,
$M=L=2$, and SNR is 20dB. For each of the scheme, the best $S$ value
was applied accordingly, which shows higher achievable rates. It is
apparent that for infinitely large $N$, the rates of all the OIA
schemes will be the same as  those of the interference-free network.
It can be seen from the figure that the SVD-based OIA with $S=1$
approaches the upper-bound most rapidly, since the interference can
be made smaller than that of the other OIA schemes according to the
given scaling laws. While the SIMO OIA is inferior  to the max-SNR
scheme if $N\le20$, both the proposed OIA schemes exhibit higher
sum-rates than those of the max-SNR scheme if $N>3$.

Finally, Fig. \ref{SER_N} illustrates the symbol error rate (SER)
averaged over all users versus $N$ when $K=3$, $M=L=2$, $S=1$, and
SNR is 20dB. The block length for each channel instance was assumed
to be 50 symbols and quadrature phase shift keying (QPSK) modulation
was used. For comparison, we considered the intercell
interference-free scheme with the random user selection, which is
labeled as `Interference-Free' in the figure. It is shown that the
SERs of all the OIA schemes approach to the SER of the
interference-free scheme as $N$ increases. The trends for the
approaching rates comply with the results of Theorem
\ref{theorem:AS} and \ref{theorem:BF}; that is, a lower user scaling
condition implies better performance, a higher approaching rate in
this case.




\section{Conclusion} \label{SEC:Conc}

We have proposed two OIA schemes for the MIMO IMAC and have derived
the user scaling law required to achieve the target $KS$ DoF.
Although the antenna selection-based OIA cannot fundamentally change
the user scaling law compared to the SIMO case, it can increase the
achievable rate even with fixed $L$ and with no feedback. Moreover,
if $L$ scales also with respect to SNR, then the scaling condition
is linearly reduced with respect to $L$. It was also shown that the
user scaling condition can be significantly reduced to
$\textrm{SNR}^{(K-1)-L+1}$ using the SVD-based OIA with help of
optimizing a beamforming vector at each user. Furthermore, the
achievable rate of the proposed OIA techniques outperform the
conventional user scheduling schemes including SIMO OIA.

From this study on the user scaling law, we characterized the lower-
and upper-bounds for the trade-off between the number of users
required to achieve a target DoF and the amount of the feedback for
the weight vectors. Even with the practical rages of the parameters,
the user scaling law is a powerful tool to analytically compare the
performance, such as the achievable rates or DoF, of any OIA schemes
for given number of users.

It can be conjectured that the MIMO OIA with limited feedback for
the weight vectors will make a bridge between the proposed two OIA
schemes. As our future work, the scaling law for the number of users
as well as  the feedback size will be studied.

\appendices
\section{Proof of Theorem \ref{theorem:AS}}\label{app:AS_theorem}
From (\ref{eq:rate_general}) and (\ref{eq:SINR_def}),
$\textrm{SINR}^{[i,j]}$ can be written as
\begin{align}
\textrm{SINR}^{[i,j]}
&= \frac{\textrm{SNR}}{\left\|\mathbf{f}_{i,j}\right\|^2 +  I_{i,j} } \\
&\ge \frac{\textrm{SNR}/\left\|\mathbf{f}_{i,j}\right\|^2}{1 +
\sum_{k=1, k\neq i}^{K} \sum_{m=1}^{S}
\left\|{\mathbf{U}_{i}}^{H}\mathbf{h}_{i,\hat{l}(k,m)}^{[k,m]}\right\|^2
\textrm{SNR} }. \label{eq:SINR_LB}
\end{align}
It is apparent that the DoF of $KS$ is achieved if the interference
term in the denominator of the right-hand side of (\ref{eq:SINR_LB})
remains constant for increasing SNR. At this point, let us define
$\mathcal{P}_{\textrm{AS}}$ by
\begin{align}
\label{eq:P_AS_def} \mathcal{P}_{\textrm{AS}} &\triangleq \lim_{\textrm{SNR}\rightarrow \infty} \textrm{Pr} \Bigg\{\sum_{k=1, k\neq i}^{K} \sum_{m=1}^{S} \left\|{\mathbf{U}}^{H}\mathbf{h}_{i,\hat{l}(k,m)}^{[k,m]}\right\|^2 \textrm{SNR}\le \epsilon,\nonumber \\
&\hspace{40pt}\forall \textrm{ user $j$ in the $i$-th cell}, i\in
\mathcal{K}, j\in \mathcal{S}  \Bigg\},
\end{align}
where $\epsilon >0$ is a positive constant. Then, DoF is bounded as
\begin{equation}  \label{eq:DoF_LB}
\textrm{DoF} \ge KS\cdot \mathcal{P}_{\textrm{AS}}.
\end{equation}
When calculating the lower bound (\ref{eq:DoF_LB}), we assumed that
the DoF of $KS$ is achieved if the interference remains constant for
increasing SNR, and zero DoF is achieved otherwise.

The essential of the OIA is the fact that the sum of the received
interference terms is equivalent to the sum of the LIF metrics of
the selected users. That is,
 \begin{align} \label{eq:LIF_equivalence}
& \sum_{i=1}^{K}\sum_{k=1, k\neq i}^{K} \sum_{m=1}^{S}
\left\|{\mathbf{U}_{i}}^{H}\mathbf{h}_{i,\hat{l}(k,m)}^{[k,m]}\right\|^2
=\sum_{i=1}^{K}\sum_{j^{\prime}=1}^{S}
\eta^{[i,j^{\prime}]}_{\textrm{AS}}.
 \end{align}
Subsequently, defining
\begin{equation}\label{eq:I_tilde}
\tilde{I}_{\textrm{AS},i} \triangleq \sum_{k=1, k\neq i}^{K}
\sum_{m=1}^{S}
\left\|{\mathbf{U}_{i}}^{H}\mathbf{h}_{i,\hat{l}(k,m)}^{[k,m]}\right\|^2,
\end{equation}
we find the following lower-bound of $\mathcal{P}_{\textrm{AS}}$:
\begin{align}
\mathcal{P}_{\textrm{AS}} &\ge \lim_{\textrm{SNR}\rightarrow \infty} \textrm{Pr} \Bigg\{\sum_{i=1}^{K}\sum_{j=1}^{S} \tilde{I}_{\textrm{AS},i} \textrm{SNR}\le \epsilon\Bigg\}\label{eq:P_AS_LB0} \\
 &= \lim_{\textrm{SNR}\rightarrow \infty} \textrm{Pr} \left\{\sum_{j=1}^{S}\sum_{i=1}^{K}\sum_{j^{\prime}=1}^{S}\eta_{\textrm{AS}}^{[i,j^{\prime}]} \textrm{SNR}\le \epsilon\right\} \label{eq:P_AS_LB1}\\
 &\ge \underbrace{\lim_{\textrm{SNR}\rightarrow \infty} \textrm{Pr} \left\{\eta_{\textrm{AS}}^{[i,j^{\prime}]} \le \frac{\textrm{SNR}^{-1}\epsilon}{KS^2}, \forall i\in \mathcal{K}, \forall j^{\prime}\in\mathcal{S}\right\}}_{\triangleq\mathcal{P}_{\textrm{AS}}^0}, \label{eq:P_AS_LB2}
\end{align}
where (\ref{eq:P_AS_LB1}) follows from (\ref{eq:LIF_equivalence}).
Unlike in the SIMO case \cite[Theorem 1]{B_Jung11_TC},
$\eta_{\textrm{AS}}^{[i, j]}$ is the minimum of $L$ independent
Chi-square random variables with degrees-of-freedom of $2(K-1)S$,
$\left\|{\mathbf{U}_{i}}^{H}\mathbf{h}_{i,l}^{[k,m]}\right\|^2$,
$l=1, \ldots, L$. We denote the probability that user $j$ in the
$i$-th cell has at least one transmit antenna with the scheduling
metric lower than $\frac{\epsilon \textrm{SNR}^{-1}}{K S^2}$ as
\begin{align}
P_a &\triangleq 1-\textrm{Pr} \Bigg\{ \sum_{k=1, k\neq i}^{K} \left\|{\mathbf{U}_k}^{H}\mathbf{h}_{k,l}^{[i,j]} \right\|^2  > \frac{\epsilon \textrm{SNR}^{-1}}{K S^2}, \nonumber \\
&\hspace{90pt}\forall l \in\{ 1, \ldots, L\} \Bigg\}.
\end{align}
It can be easily verified that $P_a$ is identical and independent
for all users. Let us denote the right-hand side of
(\ref{eq:P_AS_LB2}) \pagebreak[0] by $\mathcal{P}_{\textrm{AS}}^0$.
Note that $\mathcal{P}_{\textrm{AS}}^{0}$ represents the probability
that there exist at least $S$ users in each cell, which have the
scheduling metrics lower than $\frac{\epsilon
\textrm{SNR}^{-1}}{KS^2}$, and thus we have
\begin{equation} \label{eq:P_a_def}
\mathcal{P}_{\textrm{AS}}^0 = 1- \lim_{\textrm{SNR} \rightarrow
\infty}\sum_{i=0}^{S-1} \left( \begin{array}{c}
                                                 N \\
                                                 i
                                               \end{array}
\right) {P_a}^i \cdot (1-P_a)^{N-i}.
\end{equation}
Denoting by $F(x)$ the cumulative density function (CDF) of a
chi-square random variable with the degrees-of-freedom of $2(K-1)S$,
we have
\begin{align} \label{eq:P_a2}
P_a &= 1-  \left( 1- F\left( \frac{\epsilon
\textrm{SNR}^{-1}}{KS^2}\right) \right)^L.
\end{align}
Applying (\ref{eq:P_a2}) to (\ref{eq:P_a_def}), we get
(\ref{eq:P_OIA_AS1}) and (\ref{eq:P_OIA_AS_last}) at the bottom of
the next page,
\begin{figure*}[!b]
\hrulefill
\begin{align}
\label{eq:P_OIA_AS1}\mathcal{P}_{\textrm{AS}}^0 &= 1- \lim_{\textrm{SNR} \rightarrow \infty}\sum_{i=0}^{S-1} \frac{N!}{i!(N-i)!} \frac{\left( 1-  \left( 1- F\left( \frac{\epsilon \textrm{SNR}^{-1}}{KS^2}\right) \right)^L\right)^i \left( 1- F\left( \frac{\epsilon \textrm{SNR}^{-1}}{KS^2}\right) \right)^{LN}  }{\left( 1- F\left( \frac{\epsilon \textrm{SNR}^{-1}}{KS^2}\right) \right)^{Li}} \\
& \ge 1- \lim_{\textrm{SNR} \rightarrow \infty}\sum_{i=0}^{S-1}
\frac{ \left\{N\left( 1- \left(1- C_2 \cdot
\textrm{SNR}^{-(K-1)S}\right)^{L}   \right)\right\}^i \left( 1- C_1
\textrm{SNR}^{-(K-1)S} \right)^{LN}  }{\left( 1-
C_2\textrm{SNR}^{-(K-1)S} \right)^{Li}},\label{eq:P_OIA_AS_last}
\end{align}
\end{figure*}
where $C_1$ and $C_2$ are constants independent of SNR and $L$,
defined by
\begin{align}
C_1 = \frac{e^{-1}2^{-(K-1)S}}{(K-1)S \cdot
\Gamma((K-1)S)}\cdot\left(\frac{\epsilon}{K S^2}\right)^{(K-1)S},
\end{align}
\begin{equation}\label{eq:C_2_def}
C_2 = \frac{2^{-(K-1)S+1}}{(K-1)S\cdot \Gamma((K-1)S)} \cdot
\left(\frac{\epsilon}{K S^2}\right)^{(K-1)S}.
\end{equation}
Here, (\ref{eq:P_OIA_AS_last}) follows from the fact that
\cite[Lemma 1]{B_Jung11_Asilomar}
\begin{align}
\frac{e^{-1}2^{-(K-1)S}}{(K-1)S \cdot \Gamma((K-1)S)}\cdot x^{-(K-1)S} \le F(x), \\
F(x)\le \frac{2^{-(K-1)S+1}}{(K-1)S\cdot \Gamma((K-1)S)}\cdot
x^{-(K-1)S}
\end{align}
 and from the fact that $\frac{N!}{i!(N-i)!} \le N^i$.
Here, if we choose $\epsilon$ small enough such that $C_2
\textrm{SNR}^{-(K-1)S}<1/L$ for given SNR, we get
\begin{equation}  \label{eq:AS_C2_inequality}
\left(1- C_2 \textrm{SNR}^{-(K-1)S}\right)^L >  1- LC_2
\textrm{SNR}^{-(K-1)S},
\end{equation}
which follows from the fact that $1-xy<(1-x)^y$ for any $0<x<1<y$
and $xy\le 1$. Now, inserting (\ref{eq:AS_C2_inequality}) to
(\ref{eq:P_OIA_AS_last}) gives us
\begin{equation} \label{eq:P_AS_final}
\mathcal{P}_{\textrm{AS}}^{0} \ge 1- \lim_{\textrm{SNR} \rightarrow
\infty}\sum_{i=0}^{S-1}  \frac{ \left(N L C_2 \textrm{SNR}^{\delta}
\right)^i \left( 1- C_1 \textrm{SNR}^{\delta} \right)^{LN}  }{\left(
1- C_2\textrm{SNR}^{\delta} \right)^{Li}},
\end{equation}
where $\delta = -(K-1)S$. If $LN =
\omega\left(\textrm{SNR}^{(K-1)S}\right)$, then $\left( 1- C_1
\textrm{SNR}^{-(K-1)S} \right)^{LN}$ decreases exponentially  with
respect to SNR, whereas  $\left(N L C_2 \textrm{SNR}^{-(K-1)S}
\right)^i$ increases polynomially for any $i>0$. Therefore,
$\mathcal{P}_{\textrm{AS}}^{0}$ tends to 1 as $\textrm{SNR}$ goes to
infinity, and thereby $\mathcal{P}_{\textrm{AS}}$ tends to 1. This
proves the theorem together with (\ref{eq:DoF_LB}).

\section{Proof of Lemma \ref{lemma:F_phi}} \label{app:CDF_SVD}
Since $\mathbf{U}_k$ is chosen from an independent isotropic
distribution and $\mathbf{H}_{k}^{[i,j]}$ is an i.i.d. complex
Gaussian random matrix, for all $i, k\in \mathcal{K}$, $j
\in\mathcal{S}$, $\mathbf{G}^{[i,j]}$ is also an i.i.d. complex
Gaussian random matrix. Furthermore, both of $\mathbf{U}_k$ and
$\mathbf{H}_{k}^{[i,j]}$ are chosen from the continuous
distributions, and thus have full ranks almost surely
\cite{A_Edelman89_PhD}. The LIF metric $\eta^{[i,j]}_{\textrm{SVD}}=
{\sigma^{[i,j]}_{L}}^2$ is the smallest eigen value of the $(L
\times L)$-dimensional central  Wishart matrix
${\mathbf{G}^{[i,j]}}^{H}\mathbf{G}^{[i,j]}$. Therefore, from
\cite[Theorem 4]{S_Jin08_TC}, the polynomial CDF of the smallest
eigen value of the full-rank Wishart matrix which is constructed
from a $((K-1)S \times L)$-dimensional complex Gaussian matrix has
the smallest power of $(K-1)S-L+1$ with the multiplicative
coefficient $\alpha$ defined by \pagebreak[0]
\begin{equation}
\alpha \triangleq \frac{ \Gamma_{L-1}(1)
}{((K-1)S-L+1)!\Gamma_L(L)}\left| \boldsymbol{\Xi}\right|.
\end{equation}
Here, $\Gamma_s(t)$ is the normalized complex multivariate gamma
function, i.e., $\Gamma_s(t) = \prod_{i=1}^{s}(t-i)!$, and
$\boldsymbol{\Xi}$ is an $(L\times L)$-dimensional integer matrix
defined as
\begin{equation}
\{\boldsymbol{\Xi} \}_{i,j} \hspace{-1pt}=\hspace{-3pt}
\left\{\hspace{-3pt}\begin{array}{cc}
                                \left(\begin{array}{c}
                                   L-i \\
                                   j-i
                                 \end{array}\right)
                                 & \begin{array}{c}
                                   i=1, \ldots, L-1, j=1, \ldots, L,  \\
                                  j\ge i
                                 \end{array} \\
                                \frac{(-1)^{i-j}(L-j)!}{(n-j)!} & i=1, \ldots, L, j=1, \ldots, L, j\le i \\
                                0 & \textrm{otherwise}.
                              \end{array}
                              \right.
 \end{equation}
 Therefore, $\alpha$ is determined only by $K$, $S$, and $L$, which proves the lemma.

\section{Proof of Theorem \ref{theorem:BF}}  \label{app:BF_theorem}
From the $\textrm{SINR}^{[i,j]}$  lower bound, given by
 \begin{align}
\textrm{SINR}^{[i,j]} &\ge
 \frac{\textrm{SNR}/\left\|\mathbf{f}_{i,j}\right\|^2}{1 +  \sum_{k=1, k\neq i}^{K} \sum_{m=1}^{S} \left\|{\mathbf{U}_{i}}^{H}\mathbf{H}_{i}^{[k,m]}\mathbf{w}^{[k,m]}_{\textrm{SVD}} \right\|^2 \hspace{-3pt}\textrm{SNR} },
\end{align}
we again consider the lower bound of the DoF as
\begin{equation} \label{eq:DoF_SVD_LB}
\textrm{DoF} \ge KS \cdot \mathcal{P}_{\textrm{SVD}},
\end{equation}
\begin{align}
\label{eq:P_SVD_def}\mathcal{P}_{\textrm{SVD}} &\triangleq \lim_{\textrm{SNR}\rightarrow \infty} \textrm{Pr} \Bigg\{\tilde{I}_{\textrm{SVD},i}\textrm{SNR}\le \epsilon, \nonumber \\
& \hspace{40pt}\forall \textrm{ user $j$ in the $i$-th cell}, i\in
\mathcal{K}, j\in \mathcal{S}  \Bigg\}, \displaybreak[0]
\end{align}
where
\begin{equation}
\tilde{I}_{\textrm{SVD},i} = \sum_{k=1, k\neq i}^{K} \sum_{m=1}^{S}
\left\|{\mathbf{U}_i}^{H}\mathbf{H}_{i}^{[k,m]}\mathbf{w}_{\textrm{SVD}}^{[k,m]}
\right\|^2.
\end{equation}
Similarly to (\ref{eq:LIF_equivalence}) to  (\ref{eq:P_AS_LB2}), the
lower bound on $\mathcal{P}_{\textrm{SVD}}$ is obtained from
\begin{align}
\label{eq:P_BF_LB0}\mathcal{P}_{\textrm{SVD}}
&\ge \lim_{\textrm{SNR}\rightarrow \infty} \textrm{Pr} \left\{\sum_{i=1}^{K}\sum_{j=1}^{S}\tilde{I}_{\textrm{SVD},i} \textrm{SNR}\le \epsilon\right\} \\
 &= \lim_{\textrm{SNR}\rightarrow \infty} \textrm{Pr} \left\{\sum_{j=1}^{S}\sum_{i=1}^{K}\sum_{j^{\prime}=1}^{S}\eta_{\textrm{SVD}}^{[i,j^{\prime}]} \textrm{SNR}\le \epsilon\right\} \label{eq:P_BF_LB1}\\
 &\ge \lim_{\textrm{SNR}\rightarrow \infty} \textrm{Pr} \left\{\eta_{\textrm{SVD}}^{[i,j^{\prime}]}\le \frac{\textrm{SNR}^{-1}\epsilon}{KS^2}, \forall i\in \mathcal{K}, \forall j^{\prime}\in \mathcal{S}\right\} \label{eq:P_BF_LB2}
\end{align}
The right-hand side of (\ref{eq:P_BF_LB2}) is the probability that
there exist at least $S$ users with the scheduling metrics lower
than $\frac{\textrm{SNR}^{-1}\epsilon}{KS^2}$. Noting that
the scheduling metrics $\eta_{\textrm{SVD}}^{[i,j]}$, $i=1, \ldots,
K$, $j=1, \ldots, S$,  are identically distributed,
the right-hand side of (\ref{eq:P_BF_LB2}), denoted by
$\mathcal{P}_{\textrm{SVD}}^{0}$, can be expressed as
\begin{align}
\mathcal{P}_{\textrm{SVD}}^{0} &= 1-  \lim_{\textrm{SNR}\rightarrow
\infty} \sum_{i=0}^{S-1} \left( \begin{array}{c}
                                                                 N \\
                                                                 i
                                                               \end{array}
\right)\left(F_{\sigma}\left(\frac{\epsilon \textrm{SNR}^{-1}}{KS^2} \right)\right)^i  \nonumber \\
& \hspace{70pt}\times\left(1-F_{\sigma}\left(\frac{\epsilon
\textrm{SNR}^{-1}}{KS^2} \right)\right)^{N-i} \displaybreak[0]
\end{align}
Denoting $\rho \triangleq (K-1)S+L-1$, we further have
\begin{align}
\mathcal{P}_{\textrm{SVD}}^{0} & = 1-  \lim_{\textrm{SNR}\rightarrow \infty} \sum_{i=0}^{S-1} \frac{N!}{i!(N-i)!} \nonumber \\
& \hspace{50pt}\times\frac{\left( \Psi\textrm{SNR}^{-\rho}+o\left( \textrm{SNR}^{-\rho}\right)\right)^{i}}{\left(1-\Psi\textrm{SNR}^{-\rho}-o\left( \textrm{SNR}^{-\rho}\right)\right)^i} \nonumber \\
&\hspace{50pt}\times \left( 1-\Psi\textrm{SNR}^{-\rho}-o\left(
\textrm{SNR}^{-\rho}\right)\right)^{N} \label{eq:P_BF_LB_final2} \\
\displaybreak[0]
& \ge 1- \lim_{\textrm{SNR}\rightarrow \infty} \sum_{i=0}^{S-1} \frac{\left\{ N \left( \Psi\textrm{SNR}^{-\rho}+o\left( \textrm{SNR}^{-\rho}\right)\right)\right\}^{i}}{\left(1-\Psi\textrm{SNR}^{-\rho}-o\left( \textrm{SNR}^{-\rho}\right)\right)^i} \nonumber \\
& \hspace{50pt} \times \left( 1-\Psi\textrm{SNR}^{-\rho}-o\left(
\textrm{SNR}^{-\rho}\right)\right)^{N}\label{eq:P_BF_LB_final3}
\end{align}
where
\begin{equation}
\Psi \triangleq \alpha
\cdot\left(\frac{\epsilon}{KS^2}\right)^{(K-1)S-L+1}.
\end{equation}
Here, (\ref{eq:P_BF_LB_final2}) follows from Lemma \ref{lemma:F_phi}
and from choosing $\epsilon$ small enough such that $\frac{\epsilon
\textrm{SNR}^{-1}}{KS^2}<1$ for given $\textrm{SNR}$, and
(\ref{eq:P_BF_LB_final3}) follows from $\frac{N!}{i!(N-i)!} \le
N^i$.

Now, if $N = \omega\left(\textrm{SNR}^{\rho}\right)$,  $\left(
1-\Psi\textrm{SNR}^{-\rho}-o\left(
\textrm{SNR}^{-\rho}\right)\right)^{N}$ decreases exponentially as
$\textrm{SNR}$ increases. On the other hand, $\left\{ N \left(
\Psi\textrm{SNR}^{-\rho}+o\left(
\textrm{SNR}^{-\rho}\right)\right)\right\}^{i}$ increases
polynomially for any $i>0$, and thus, the second term of
(\ref{eq:P_BF_LB_final3}) tends to zero as $\textrm{SNR} \rightarrow
\infty$. Therefore, the lower bound of $\mathcal{P}_{\textrm{SVD}}$
given in (\ref{eq:P_BF_LB2}) tends to 1, which proves the theorem
together with (\ref{eq:DoF_SVD_LB}).

\section{Proof of Theorem \ref{th:different_parameters}} \label{app:cell_dependent}
Following (\ref{eq:P_AS_LB0}) to (\ref{eq:P_AS_LB2}) and
(\ref{eq:P_BF_LB0}) to (\ref{eq:P_BF_LB2}) and denoting the scheme
indicator by $\tau \in \{\textrm{AS}, \textrm{SVD} \}$,
$\mathcal{P}_{\tau}$ with the cell-dependent parameters can be
written by
\begin{align}
\mathcal{P}_{\tau}&\ge \lim_{\textrm{SNR}\rightarrow \infty} \textrm{Pr} \left\{\sum_{i=1}^{K}\sum_{j=1}^{S_i}\sum_{j^{\prime}=1}^{S_i}\eta_{\tau}^{[i,j^{\prime}]} \textrm{SNR}\le \epsilon\right\} \label{eq:P_tau_LB1} \\
& \ge \lim_{\textrm{SNR}\rightarrow \infty} \textrm{Pr} \left\{\eta_{\tau}^{[i,j^{\prime}]}\le \frac{\textrm{SNR}^{-1}\epsilon}{\sum_{i^{\prime}=1}^{K}{S_{i^{\prime}}}^2}, \forall i\in \mathcal{K}, \forall j^{\prime}\in \mathcal{S}\right\} \\
\label{eq:P_tau_LB2}& = \lim_{\textrm{SNR}\rightarrow \infty}
\prod_{i=1}^{K}\underbrace{\textrm{Pr}
\left\{\eta_{\tau}^{[i,j^{\prime}]}\le
\frac{\textrm{SNR}^{-1}\epsilon}{\sum_{i^{\prime}=1}^{K}{S_{i^{\prime}}}^2},
\forall j^{\prime}\in \mathcal{S}\right\}}_{\triangleq
\mathcal{P}_{\tau}^{[i]}},
\end{align}
where 
in (\ref{eq:P_tau_LB2}), $\mathcal{P}_{\tau}^{[i]}$ denote the
probability there exist at least $S_i$ users with LIF metrics
smaller than
$\frac{\textrm{SNR}^{-1}\epsilon}{\sum_{i^{\prime}=1}^{K}{S_{i^{\prime}}}^2}$
at the $i$-th cell, which is independent from those of the other
cells.

i) Antenna selection-based OIA: Since
$\left\|{\mathbf{U}_k}^{H}\mathbf{h}_{k,\hat{l}(i,j)}^{[i,j]}
\right\|^2$ is a Chi-square random variable with DoF of $2S_k$, the
scheduling metric $\eta^{[i,j]}_{\textrm{AS}}$ in
(\ref{eq:AS_select_LIF_metric}) is a Chi-square random variable with
DoF of $2S^{\prime}$, where $S^{\prime} = 2\sum_{k\neq i, k=1}^{K}
S_k$. The rest of the proof can be done analogously to the proof for
Theorem \ref{theorem:AS} replacing $(K-1)S$ with $S^{\prime}$.

ii) SVD-based OIA: Since
${\mathbf{U}_{k}}^{H}\mathbf{H}_{k}^{[i,j]}$ is an $(S_k \times
L_i)$ dimensional Gaussian matrix, $\mathbf{G}^{[i,j]}$ defined in
(\ref{eq:G_def}) is now $(S^{\prime} \times L_i)$-dimensional.
Following the analogous derivation of the proof for Theorem
\ref{theorem:BF} and replacing $(K-1)S$ with $S^{\prime}$, we can
complete the proof.



\newpage

\begin{figure}[t!]
  \begin{center}
  \leavevmode \epsfxsize=0.49\textwidth   
  \leavevmode 
  \epsffile{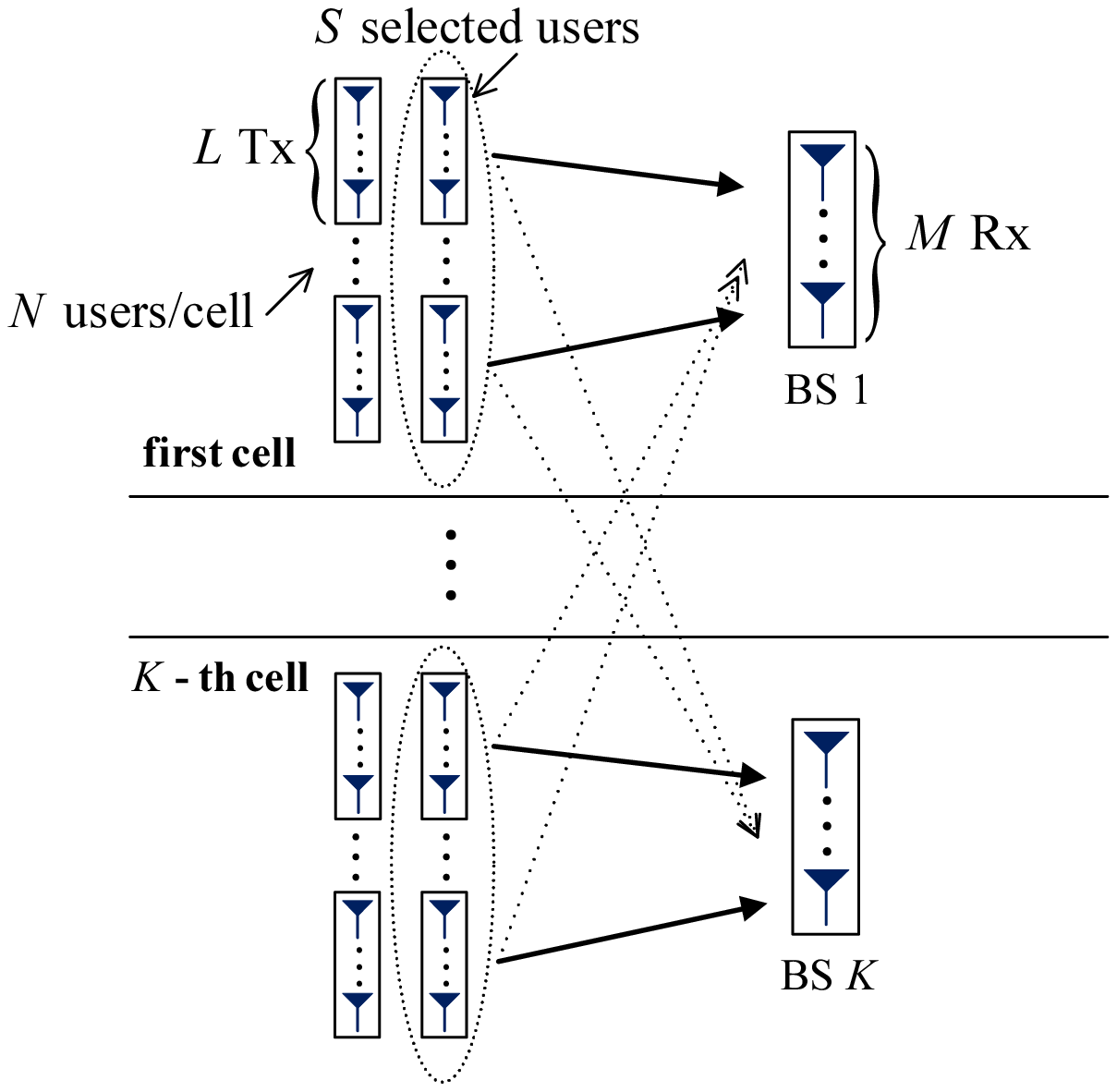}
  \caption{$K$-cell MIMO Interfering MAC.}
  \label{fig_sys}
  \end{center}
\end{figure}

\begin{figure}[t!]
  \begin{center}
  \leavevmode \epsfxsize=0.88\textwidth   
  \leavevmode 
  \epsffile{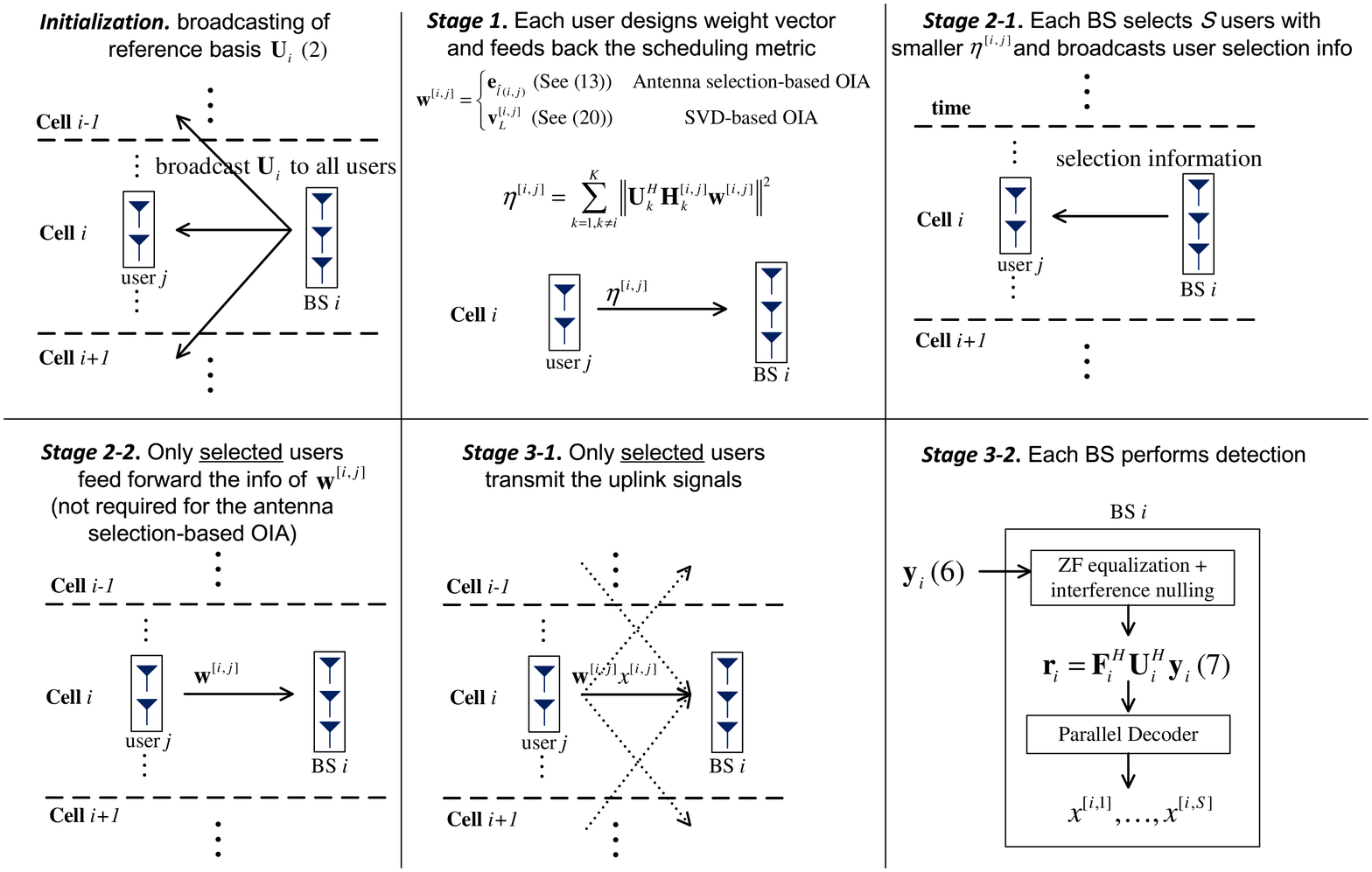}
  \caption{Overall sequential procedure of the proposed MIMO OIA.}
  \label{fig_block}
  \end{center}
\end{figure}

\begin{figure}[t!]
  \begin{center}
  \leavevmode \epsfxsize=0.61\textwidth   
  \leavevmode 
  \epsffile{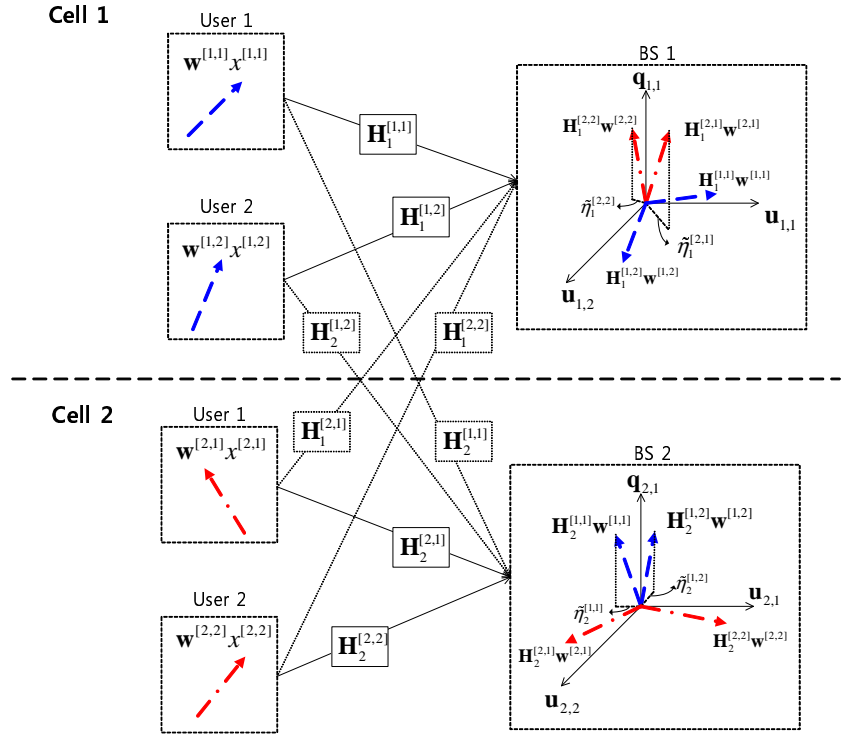}
  \caption{Proposed MIMO OIA where $K=2$, $M=3$, and $S=2$.}
  \label{fig_OIA}
  \end{center}
\end{figure}

\begin{table*}\label{table:complexity}
\caption{Computational complexity of the OIA schemes (flops).}
\begin{center}
\begin{tabular}{|c|c|c|c|}
  \hline
             & SIMO OIA \cite{B_Jung11_Asilomar} & Antenna selection-based OIA & SVD-based OIA \\\hline
  User & $O(KMS)$  & $O(KLMS)$ & $O(KSL^2 + KLMS)$\\\hline
  BS& $O(N+MS^2)$ & $O(N+MS^2)$ &  $O(N + MLS + MS^2)$\\
  \hline
\end{tabular}
\end{center}
\end{table*}

\begin{figure}[t!]
  \begin{center}
  \leavevmode \epsfxsize=0.64\textwidth   
  \leavevmode 
  \epsffile{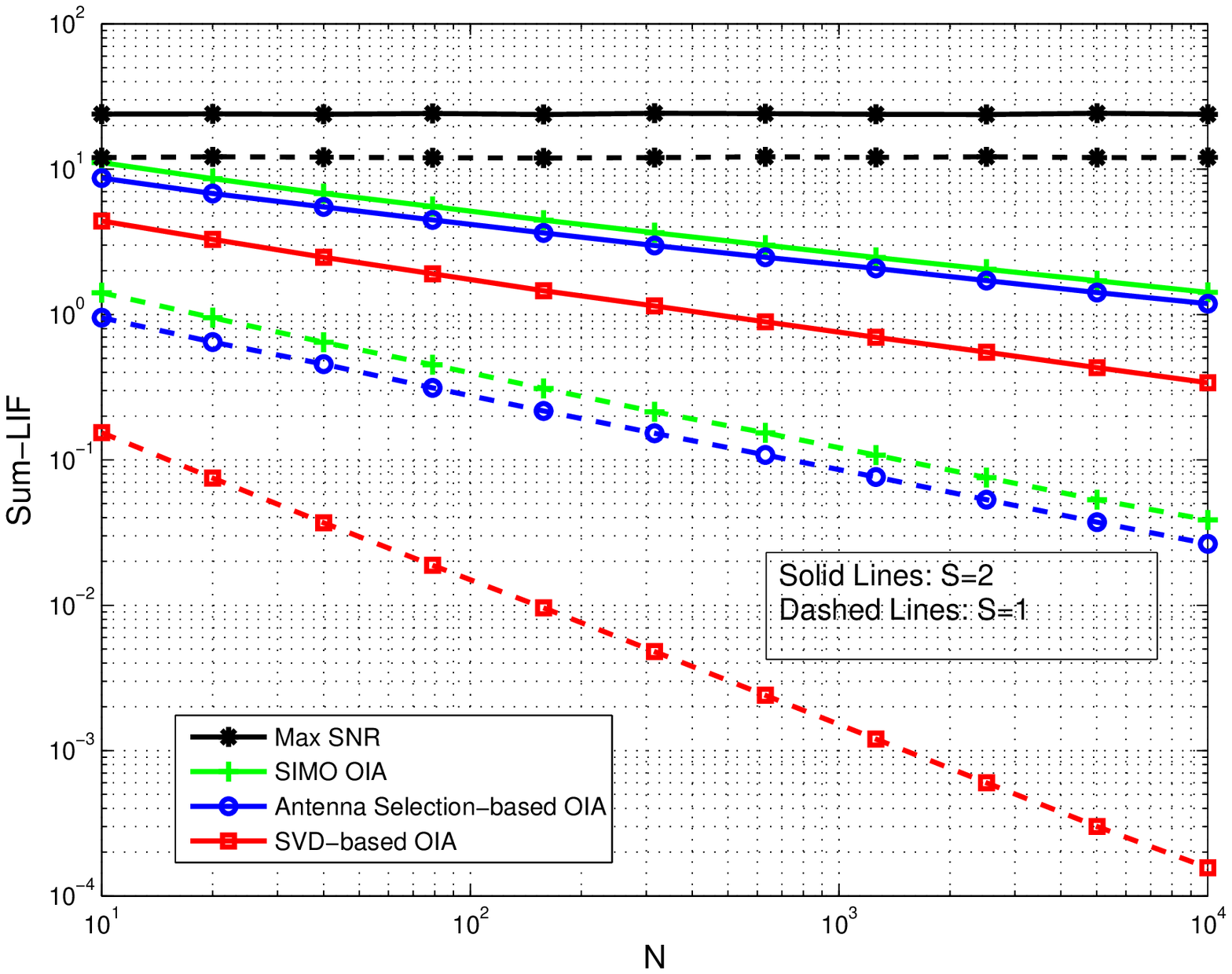}
  \caption{Sum-LIF versus $N$ for the MIMO IMAC with $K=3$ and $M=L=2$.}
  \label{LIF_N}
  \end{center}
\end{figure}

\begin{figure}[t!]
  \begin{center}
  \leavevmode \epsfxsize=0.64\textwidth   
  \leavevmode 
  \epsffile{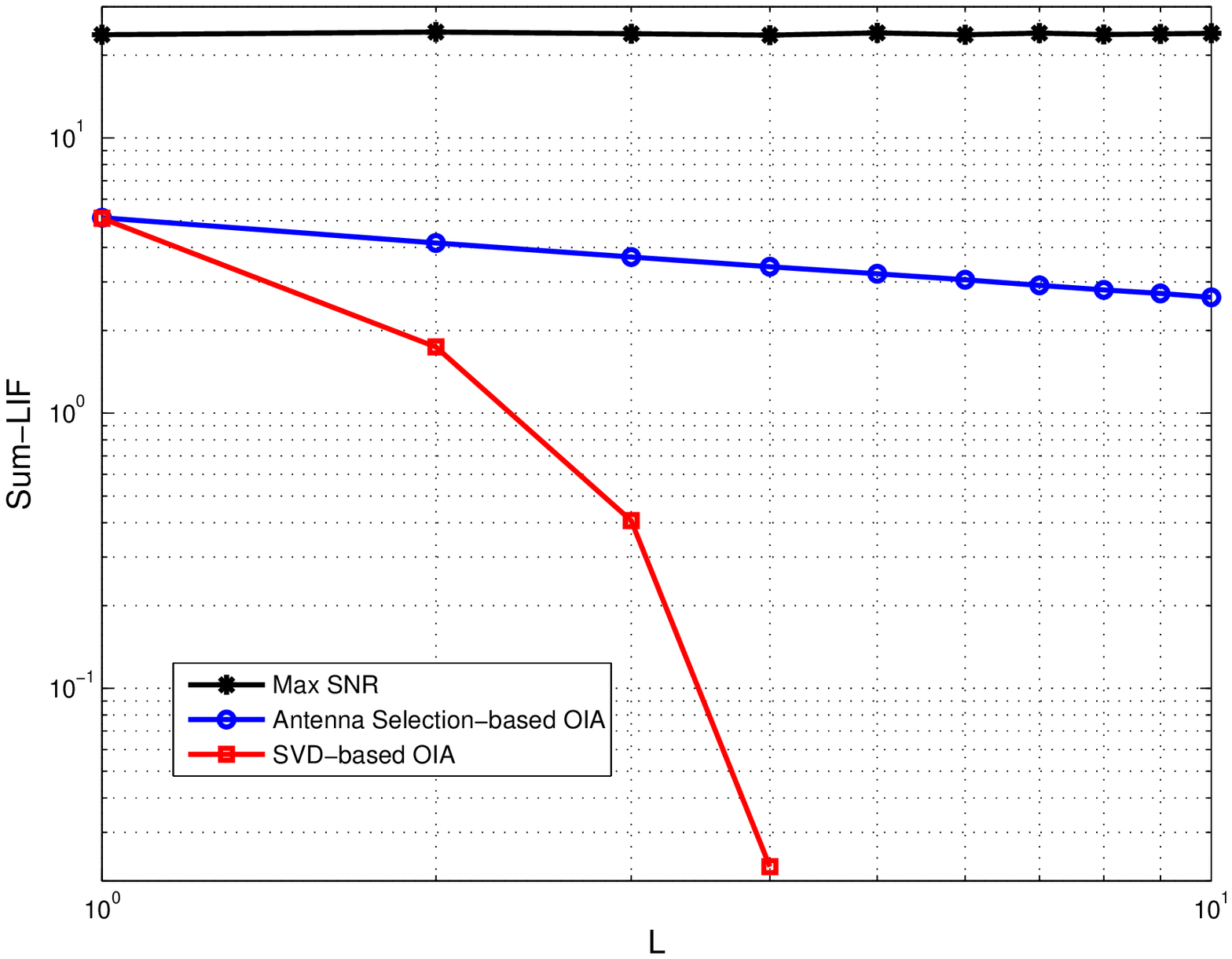}
  \caption{Sum-LIF versus $L$ for the MIMO IMAC with $K=3$, $M=2$, and $N=100$.}
  \label{LIF_L}
  \end{center}
\end{figure}

\begin{figure}[t!]
\centering{%
\subfigure[\!\!\!\!\!\!\!\!\!\!\!]{%
\epsfxsize=0.53\textwidth \leavevmode \label{rates_SNR_N100}
\epsffile{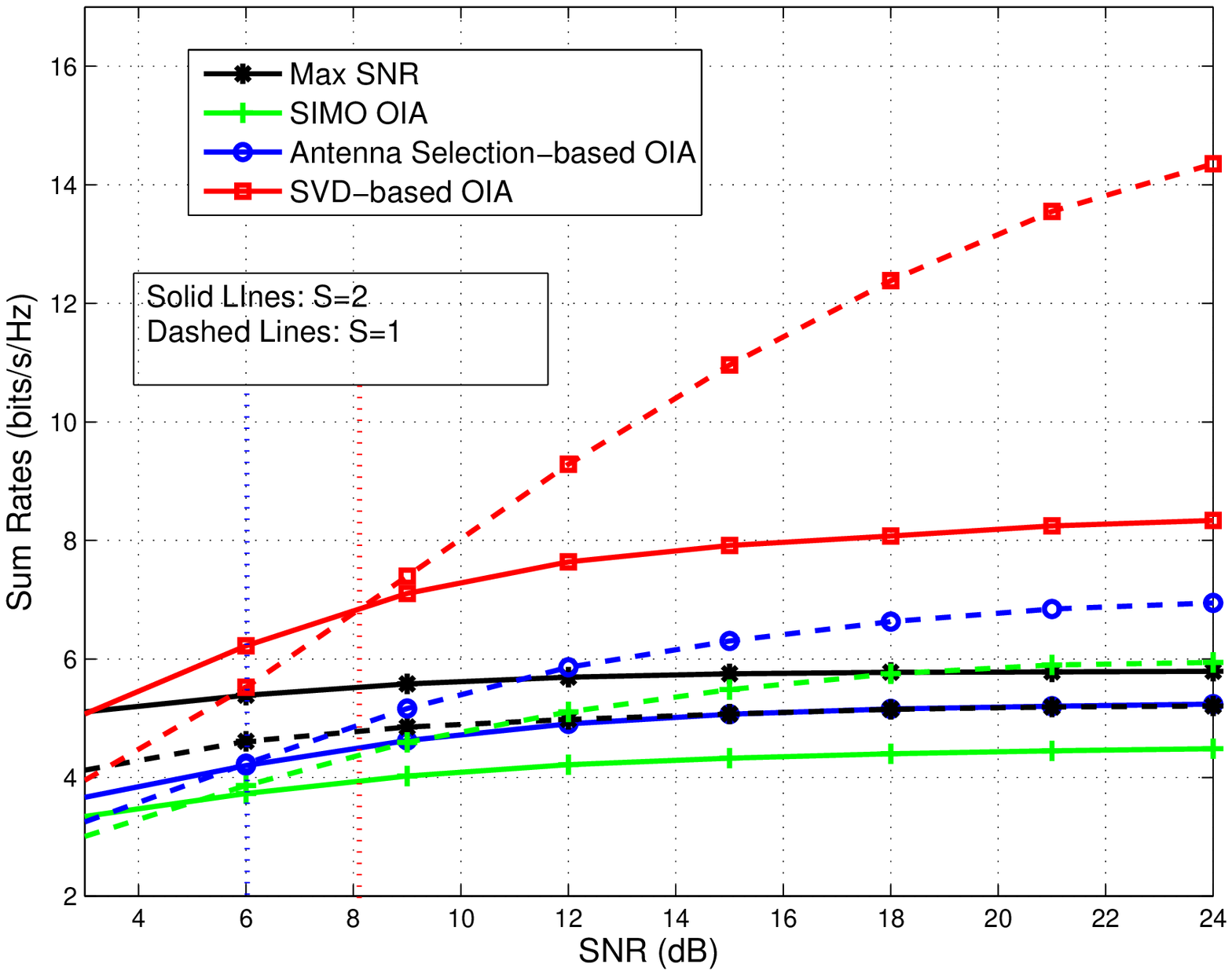}}\vspace{.0cm}
\subfigure[\!\!\!\!\!\!\!\!\!\!\!]{%
\epsfxsize=0.53\textwidth \leavevmode 
\epsffile{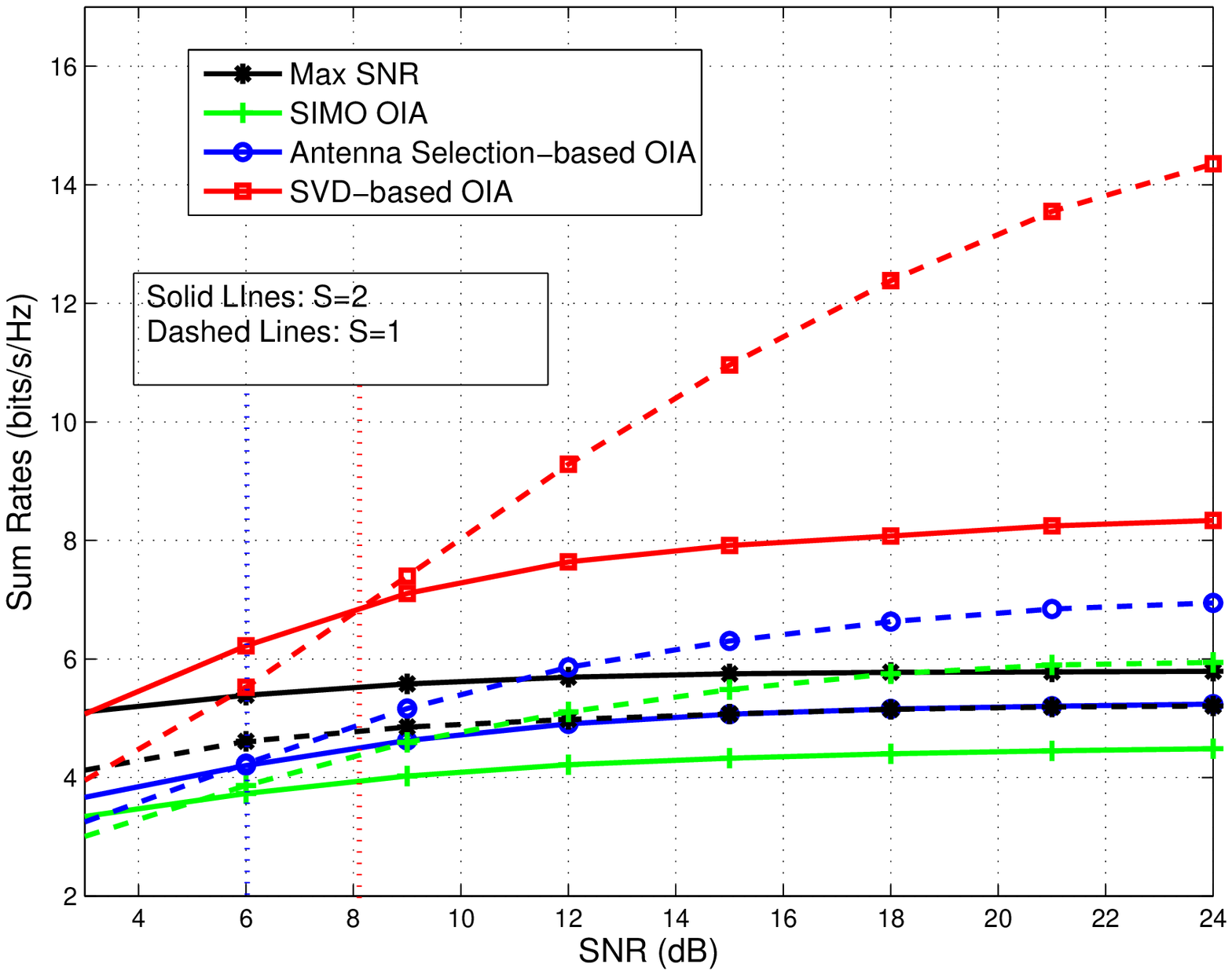}}} \caption{Achievable sum-rates versus
SNR when $K=3$, $M=L=2$, and (a) $N=20$ and (b) $N=100$.}
\label{rates_SNR}
\end{figure}

\begin{figure}[t!]
  \begin{center}
  \leavevmode \epsfxsize=0.64\textwidth   
  \leavevmode 
  \epsffile{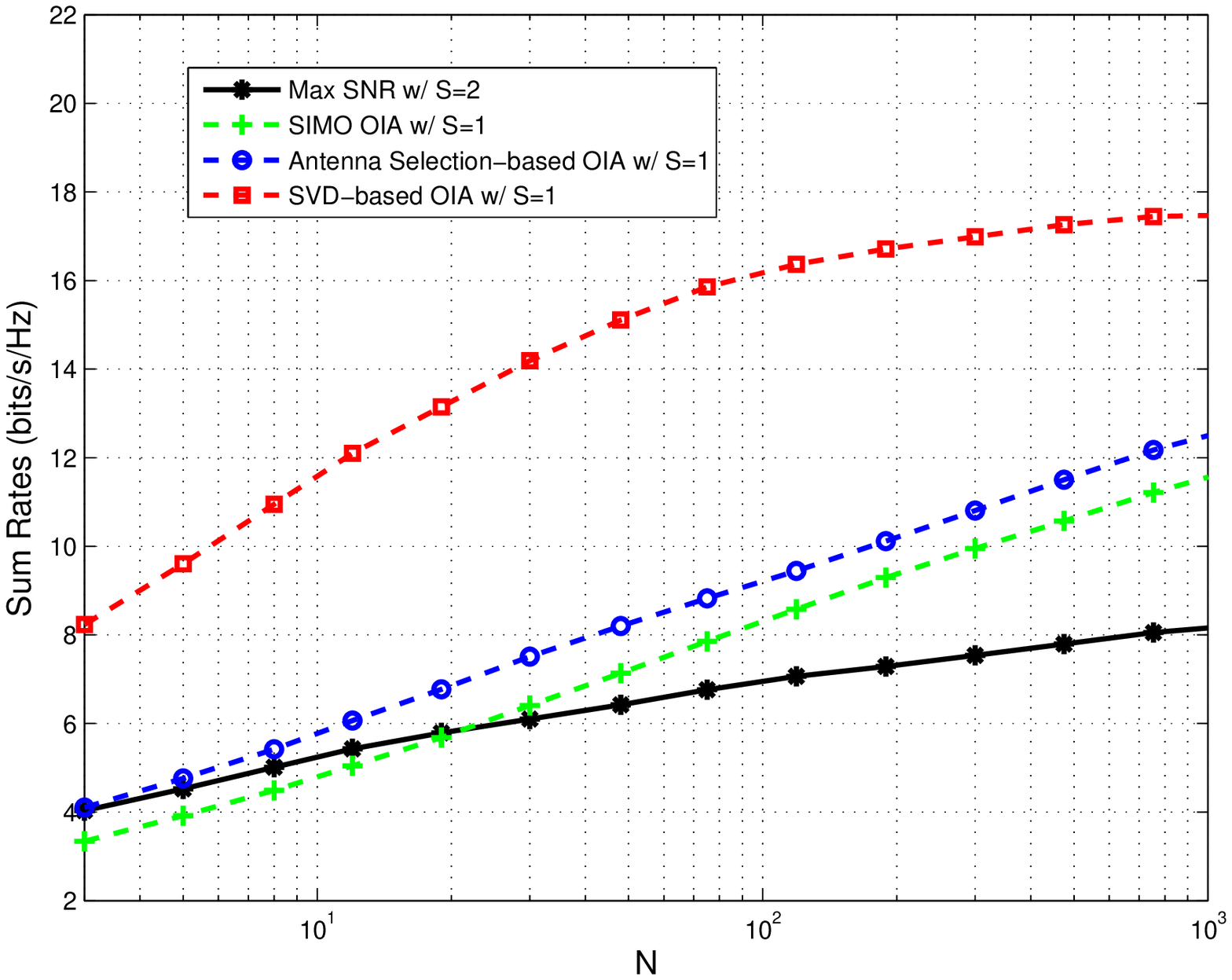}
  \caption{Achievable sum-rates versus $N$ when $K=3$, $M=L=2$, and SNR=20dB.}
  \label{rates_N}
  \end{center}
\end{figure}

\begin{figure}[t!]
  \begin{center}
  \leavevmode \epsfxsize=0.64\textwidth   
  \leavevmode 
  \epsffile{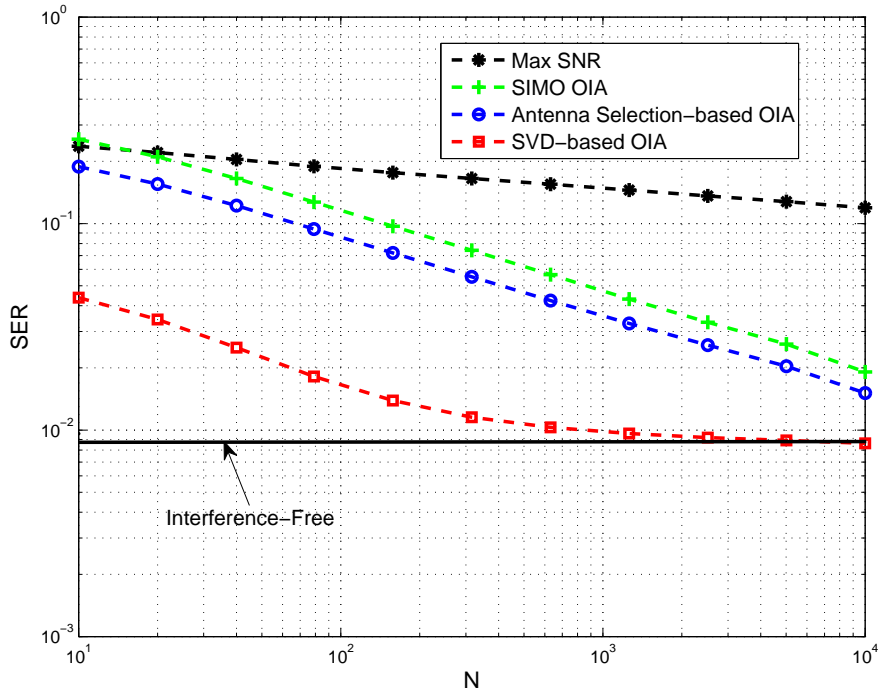}
  \caption{Average SER versus $N$ with QPSK signaling when $K=3$, $M=L=2$, $S=1$, and SNR=20dB.}
  \label{SER_N}
  \end{center}
\end{figure}

\end{document}